\documentclass[11pt]{article}

\usepackage{makecell}
\usepackage{mathrsfs}
\usepackage{amscd}
\usepackage{amsmath}
\usepackage{amssymb}
\usepackage{amstext}
\usepackage{amsthm}
\usepackage{bbold}
\usepackage{bm}
\usepackage{booktabs}
\usepackage{color}
\usepackage{easybmat}
\usepackage{framed}
\usepackage[dvips,letterpaper,margin=1in]{geometry}
\usepackage{graphicx}
\usepackage{hyperref}
\usepackage[noabbrev,capitalize]{cleveref}
\usepackage{mathtools}
\usepackage{colonequals}
\usepackage{longtable}
\usepackage{rotating}
\usepackage{setspace}
\usepackage{tabu}
\usepackage{verbatim}
\usepackage[dvipsnames]{xcolor}

\colorlet{mylinkcolor}{BurntOrange}
\colorlet{mycitecolor}{Green}
\colorlet{myurlcolor}{Blue}

\hypersetup{
  linkcolor  = mylinkcolor,
  citecolor  = mycitecolor,
  urlcolor   = myurlcolor,
  colorlinks = true,
}

\newtheorem{theorem}{Theorem}[section]
\newtheorem{remark}[theorem]{Remark}
\newtheorem{lemma}[theorem]{Lemma}
\newtheorem{claim}[theorem]{Claim}
\newtheorem{question}[theorem]{Question}

\newtheorem{definition}[theorem]{Definition}

\newtheorem{proposition}[theorem]{Proposition}
\newtheorem{corollary}[theorem]{Corollary}
\newtheorem{conjecture}[theorem]{Conjecture}

\crefname{theorem}{Theorem}{Theorems}
\crefname{proposition}{Proposition}{Propositions}
\crefname{lemma}{Lemma}{Lemmas}

\theoremstyle{plain} % just in case the style had changed
\newcommand{\thistheoremname}{}
\newtheorem*{genericthm}{\thistheoremname}

\colorlet{shadecolor}{gray!25}
{\endMakeFramed}

\DeclareSymbolFont{bbold}{U}{bbold}{m}{n}
\DeclareSymbolFontAlphabet{\mathbbold}{bbold}

\newcommand{\ra}{\rangle}
\newcommand{\la}{\langle}

\makeatletter
\def\moverlay{\mathpalette\mov@rlay}
\def\mov@rlay#1#2{\leavevmode\vtop{%
   \baselineskip\z@skip \lineskiplimit-\maxdimen
   \ialign{\hfil$\m@th#1##$\hfil\cr#2\crcr}}}
\newcommand{\charfusion}[3][\mathord]{
    #1{\ifx#1\mathop\vphantom{#2}\fi
        \mathpalette\mov@rlay{#2\cr#3}
      }
    \ifx#1\mathop\expandafter\displaylimits\fi}
\makeatother

\newcommand{\RR}{\mathbb{R}}

\newcommand{\NN}{\mathbb{N}}

\newcommand{\PP}{\mathbb{P}}
\newcommand{\EE}{\mathbb{E}}

\renewcommand{\SS}{\mathbb{S}}

\newcommand{\One}{\mathbbold{1}}

\newcommand{\what}{\widehat}

\newcommand{\one}{\bm{1}}

\newcommand{\sG}{\mathcal{G}}

\newcommand{\sN}{\mathcal{N}}
\newcommand{\sO}{\mathcal{O}}

\newcommand{\Tr}{\mathrm{Tr}}

\newcommand{\eqd}{\stackrel{\mathrm{(d)}}{=}}

\newcommand{\SOS}{\mathsf{SOS}}

\newcommand{\M}{\mathsf{M}}

\newcommand{\sym}{\mathsf{sym}}

\newcommand{\diag}{\mathsf{diag}}
\newcommand{\op}{\mathsf{op}}
\newcommand{\rank}{\mathsf{rank}}

\newcommand{\GOE}{\mathsf{GOE}}

\newcommand{\bA}{\bm A}
\newcommand{\bB}{\bm B}

\newcommand{\bD}{\bm D}

\newcommand{\bF}{\bm F}
\newcommand{\bG}{\bm G}

\newcommand{\bM}{\bm M}

\newcommand{\bP}{\bm P}
\newcommand{\bQ}{\bm Q}
\newcommand{\bR}{\bm R}
\newcommand{\bS}{\bm S}
\newcommand{\bT}{\bm T}

\newcommand{\bV}{\bm V}
\newcommand{\bW}{\bm W}
\newcommand{\bX}{\bm X}
\newcommand{\bY}{\bm Y}
\newcommand{\bZ}{\bm Z}
\newcommand{\ba}{\bm a}

\newcommand{\be}{\bm e}
\newcommand{\bg}{\bm g}
\newcommand{\bh}{\bm h}
\newcommand{\bi}{\bm i}

\newcommand{\bv}{\bm v}
\newcommand{\bw}{\bm w}
\newcommand{\bbm}{\bm m}

\newcommand{\bx}{\bm x}

\newcommand{\by}{\bm y}
\newcommand{\bz}{\bm z}

\newcommand{\Haar}{\mathsf{Haar}}

\newcommand{\offdiag}{\mathsf{offdiag}}

\newcommand{\isovec}{\mathsf{isovec}}

\newcommand{\orth}{{(\mathsf{orth})}}
\newcommand{\norm}{{(\mathsf{norm})}}

\newcommand{\fC}{\mathscr{C}}
\newcommand{\fE}{\mathscr{E}}

\renewcommand{\emptyset}{\varnothing}

\DeclareRobustCommand{\bmrob}[1]{\bm{#1}}
\title{A Tight Degree 4 Sum-of-Squares Lower Bound for the Sherrington-Kirkpatrick Hamiltonian}

\usepackage{authblk}
\author[1]{%
Dmitriy Kunisky\thanks{ Email:~\textit{kunisky@cims.nyu.edu}. Partially supported by NSF grants DMS-1712730 and DMS-1719545.}}
\author[1,2]{%
  Afonso S.\ Bandeira\thanks{ Email:~\textit{bandeira@cims.nyu.edu}. Partially supported by NSF grants DMS-1712730 and DMS-1719545, and by a grant from the Sloan Foundation.}}

\affil[1]{\normalsize Department of Mathematics, Courant Institute of Mathematical Sciences, New York University}
\affil[2]{\normalsize Center for Data Science, New York University}

\date{
  First Draft: July 26, 2019 \\
  Current Draft: August 31, 2020}

\begin{document}

\pagenumbering{gobble}

\maketitle
\begin{abstract}
    We show that, if $\bW \in \RR^{N \times N}_{\sym}$ is drawn from the gaussian orthogonal ensemble, then with high probability the degree 4 sum-of-squares relaxation cannot certify an upper bound on the objective $N^{-1} \cdot \bx^\top \bW \bx$ under the constraints $x_i^2 - 1 = 0$ (i.e.\ $\bx \in \{ \pm 1 \}^N$) that is asymptotically smaller than $\lambda_{\max}(\bW) \approx 2$. We also conjecture a proof technique for lower bounds against sum-of-squares relaxations of any degree held constant as $N \to \infty$, by proposing an approximate pseudomoment construction.
\end{abstract}

\newpage

{
\hypersetup{linkcolor=black}
\tableofcontents
}

\newpage

\pagenumbering{arabic}

\section{Introduction}

\subsection{Algorithms for the Sherrington-Kirkpatrick Hamiltonian}

This paper concerns convex relaxations of the following optimization problem:
\begin{equation}
    \label{eq:M-def}
    \mathsf{M}(\bW) \colonequals \frac{1}{N}\max_{\bx \in \{\pm 1\}^N}\bx^\top \bW \bx.
\end{equation}
Since the constraint $x_i \in \{ \pm 1\}$ may be written $x_i^2 - 1 = 0$, this is a simple instance of \emph{quadratically constrained quadratic programming}.
We are moreover interested in a random setting, where $\bW \in \RR^{N \times N}_{\sym}$ is a random matrix drawn from the \emph{gaussian orthogonal ensemble (GOE)}: $W_{ii} \sim \sN(0, 2/N)$ and $W_{ij} = W_{ji} \sim \sN(0, 1/N)$, with the $N(N + 1) / 2$ entries on and above the diagonal distributed independently.
We denote this distribution $\bW \sim \GOE(N)$.
Under this model, the spectral radius of $\bW$ is of constant order, and the normalization in \eqref{eq:M-def} is such that $\EE \mathsf{M}(\bW)$ also remains of constant order as $N \to \infty$, as we will describe below.

The problem $\mathsf{M}(\bW)$ for general $\bW$ includes the problem of finding maximum cuts in graphs (MaxCut), when $\bW$ is taken to be a graph Laplacian.
Karp's classical result \cite{Karp-NP} therefore implies that computing $\mathsf{M}(\bW)$ is $\mathsf{NP}$-hard in the worst case.
The case $\bW \sim \GOE(N)$ is a simple and mathematically elegant example with which we hope to probe the average-case complexity of the same problem, seeking to understand whether the worst-case complexity abates for specific random models of $\bW$.

We are assisted in this task by the rich history of the random optimization problem $\mathsf{M}(\bW)$ in statistical physics: up to a change in sign, its value is the ground-state energy of the \emph{Sherrington-Kirkpatrick (SK) model}, a prominent mean-field model of spin glasses \cite{SK}.
In particular, the asymptotics of its expected value have been well-understood at a non-rigorous level since the seminal work of Parisi \cite{Parisi-SK}, who developed a system of deep conjectures on the optimization landscape of $\mathsf{M}(\bW)$, which, among other results, allowed him to analytically predict the limit
\begin{equation}
    \label{eq:P-star}
    \lim_{N \to \infty} \EE \mathsf{M}(\bW) \equalscolon 2\mathsf{P}_* \approx 1.5264.
\end{equation}
(Standard results from general gaussian process theory also imply strong concentration around the expectation.)
More recently, the computation of this limit has been made mathematically rigorous  as well \cite{Panchenko-ultra,Panchenko-SK,Talagrand-Parisi}.

From the perspective of computer science and optimization, perhaps the more natural random model of $\mathsf{M}(\bW)$ is the case where $\bW$ is the adjacency matrix or graph Laplacian of a random graph, which gives randomized instances of MaxCut.
A pair of elegant recent works \cite{MS-deg2,DMS-cut} showed that, in fact, for sparse random graphs this problem is intimately related to the gaussian setting of the SK model: an interpolation argument may be used to control both the true value and the value of a certain simple semidefinite programming relaxation of $\mathsf{M}(\bW)$ for sparse random graphs in terms of the SK model.

Thus, whether motivated by the mathematical interest of the GOE and SK model or the application to MaxCut, we are led to ask:
\begin{question}
    Under $\bW \sim \GOE(N)$, can $\mathsf{M}(\bW)$ be approximated accurately and efficiently?
\end{question}
Of course, knowing the limiting expectation \eqref{eq:P-star} and concentration around this value, it is simple to produce a vacuous algorithm that outputs the value $2\mathsf{P}_*$.
To capture the difficulty of solving instances of $\mathsf{M}(\bW)$ for specific random draws of $\bW$, we must therefore refine our question.

One way to do this is to ask instead:
\begin{question}
    Under $\bW \sim \GOE(N)$, can $\bx = \bx(\bW) \in \{\pm 1\}^N$ be efficiently computed such that $\frac{1}{N}\bx^\top \bW \bx \approx 2\mathsf{P}_*$?
\end{question}
Recently, assuming a widely-believed conjecture from the spin glass literature, Montanari answered this question in the affirmative in \cite{Montanari-SK}.
Montanari's result followed a similar one on local search in a simpler ``random energy model''~\cite{ABM-rem}, and used proof techniques related to those proposed by Subag in \cite{Subag-frsb}, who addressed the same question in a continuous setting.
\begin{theorem}[Theorem 2 of \cite{Montanari-SK}]
    Conditional on the conjecture that the Parisi distribution has continuous support at sufficiently low temperature in the SK model,\footnote{See Assumption 1 of \cite{Montanari-SK} and the surrounding citations and discussion for further details.} for any $\epsilon > 0$, there is a polynomial-time algorithm computing $\bx = \bx(\bW) \in \{ \pm 1\}^N$ such that
    \begin{equation}
    \lim_{N \to \infty}\PP\left[\frac{1}{N} \bx^\top \bW \bx \geq 2\mathsf{P}_* - \epsilon\right] = 1.
    \end{equation}
\end{theorem}

Another way to refine our question is to ask rather for \emph{certificates} of upper bounds on $\mathsf{M}(\bW)$:
\begin{question}
    Can $c(\bW) \in \RR$ be efficiently computed with $c(\bW) \geq \mathsf{M}(\bW)$ and $c(\bW) \approx 2\mathsf{P}_*$?
\end{question}
(Note that we require $c(\bW) \geq \mathsf{M}(\bW)$ to hold \emph{for every $\bW$}; the algorithm is not allowed to ``cheat'' the random setting by merely outputting a number slightly larger than $2\mathsf{P}_*$.)
One simple but sub-optimal approach is to form the \emph{spectral certificate}, which amounts to disregarding the constraint $\bx \in \{\pm 1\}^N$ by taking $c(\bW) \colonequals \lambda_{\max}(\bW) \approx 2$.
Recently, Montanari asked\footnote{The authors learned of this problem through private communications soon after \cite{MS-deg2} was published. More recently, it was also included in the problem list ``AimPL: Phase transitions in randomized computational problems,'' available online at \texttt{http://aimpl.org/phaserandom}.} whether any certification algorithm could improve on this performance, a problem which, besides modest progress that we will review in the following sections, has since remained open to the best of our knowledge.

Our contribution in this paper is to provide evidence that the spectral certificate is asymptotically optimal by showing that the degree 4 sum-of-squares relaxation, a much more sophisticated convex relaxation, achieves the same performance.

\subsection{Conjectural hardness of certification}

One step towards making a convincing prediction of whether better-than-spectral certification is possible in the SK model was taken in \cite{BKW-LDLR}, in which the authors participated.
In this work, we first showed that, if efficient certification below 2 were possible for the SK model, then it would be possible to efficiently perform a certain hypothesis testing task in a variant of a spiked matrix model.
Then, we provided evidence that this hypothesis testing task should be hard using a method based on the \emph{low-degree likelihood ratio}.
Roughly speaking, this technique takes low-degree polynomials as a proxy for all polynomial-time testing statistics and measures their performance in a convenient smoothed sense, which allows the optimal low-degree polynomial statistic to be identified and analyzed using an orthogonal polynomial decomposition.

This suggests the following conjecture, which would hold conditional on another, quite broad conjecture of \cite{HS-bayesian,sam-thesis} that the low-degree likelihood ratio analysis is correct for a large class of hypothesis testing problems.
\begin{conjecture}
    \label{conj:cert}
    For any $\epsilon > 0$, there does not exist a polynomial-time certification algorithm for $\M(\bW)$ such that $c(\bW) \leq 2 - \epsilon$ with high probability.
\end{conjecture}
\noindent
Unfortunately, though the low-degree likelihood ratio method predicts many known computational thresholds in random problems correctly, at the moment it is only known to imply rather weak lower bounds against specific algorithms---either only lower bounds in expectation or a smoothed $L^2$ sense, or high-probability lower bounds under quite restrictive assumptions (see, e.g., the recent survey \cite{low-deg-notes} by the authors).
In search of further evidence of hardness of certification, we therefore consider concrete algorithms and analyze their performance directly.

\subsection{Basic notions of sum-of-squares relaxations}

The main algorithmic approach for certifying bounds on a problem like $\mathsf{M}(\bW)$ is to form \emph{convex relaxations} that may be solved efficiently by standard convex optimization techniques.
First, note that, defining the \emph{cut polytope}
\begin{equation}
    \fC^N \colonequals \mathsf{conv}\left(\left\{ \bx\bx^\top: \bx \in \{ \pm 1\}^N\right\}\right),
\end{equation}
we may rewrite $\mathsf{M}(\bW)$ as a linear optimization problem over this set of matrices,
\begin{equation}
    \mathsf{M}(\bW) = \frac{1}{N}\max_{\bM \in \fC^N} \la \bW, \bM \ra.
\end{equation}
Though $\fC^N$ is a convex set, it is complex to describe \cite{DL-cut}, and in particular does not admit a polynomial-time separation oracle unless $\mathsf{P} = \mathsf{NP}$ (by the same result of \cite{Karp-NP} mentioned before).
We thus pursue the idea of expanding $\fC^N$ to a larger convex set that may be described more simply, and over which convex optimization is tractable.

Specifically, we will study the performance of \emph{semidefinite programming (SDP)} relaxations of $\mathsf{M}(\bW)$.
Perhaps the simplest of these is based on the inclusion of sets
\begin{equation}
    \fC^N \subseteq \{\bM \in \RR^{N \times N}_{\sym}: \bM \succeq \bm 0, \Tr(\bM) = N\} \equalscolon \fE^N_{\mathsf{spec}}.
\end{equation}
Replacing $\fC^N$ with $\fE^N_{\mathsf{spec}}$ in the definition of $\mathsf{M}(\bW)$ just computes $\lambda_{\max}(\bW)$, expressing the spectral certificate as an SDP relaxation.
Thus one consequence of Conjecture~\ref{conj:cert} is that this naive relaxation is optimal among the wide variety of SDP relaxations that may be applied to $\mathsf{M}(\bW)$.

One broad and successful framework for SDP relaxation of optimization problems through which one might hope to find an improvement is the \emph{sum-of-squares (SOS) hierarchy} of relaxations \cite{Lasserre-SOS-survey,Laurent-SOS-survey,BPT-book,BS-SOS}.
This generates a sequence of convex sets we will denote by $\fE_d^N$, indexed by a parameter $d$, an even natural number called the \emph{degree}, which satisfy \cite{Laurent-SOS,Fawzi-SOS} the strict inclusions
\begin{equation}
    \fE_2^N \supset \fE_4^N \supset \cdots \supset \fE_{N + \One\{N \text{ odd}\}} = \fC^N.
\end{equation}
Moreover, $\fE_d^N$ is a projection of an affine slice of the positive semidefinite cone of $N^{d / 2} \times N^{d / 2}$ real symmetric matrices, whereby optimization over $\fE_d^N$ may be written as an SDP.

Unfortunately, though in this SDP both the dimension of the decision variable and the number of constraints scale polynomially with $N$, and even if we assume (as will be the case in our application) that the SDP coefficients may be well-approximated with an encoding in polynomially many bits (i.e., are of size bounded by $\exp(\mathsf{poly}(N))$), it still need not be the case that the SDP can be solved to small additive error in polynomial time, as O'Donnell has pointed out \cite{ODonnell-2017-SOSNotAutomatizable}.
The key point is that one must further ensure that there exists an \emph{optimizer} of the SDP whose entries are also of size bounded by $\exp(\mathsf{poly}(N))$.
Fortunately, the work \cite{RW-2017-BitComplexity} studied this issue for sum-of-squares relaxations of many discrete problems, including our setting of unconstrained optimization over Boolean variables, and showed that this condition is in fact satisfied (see their Corollary 9).
Thus optimization over $\fE_d^N$ may indeed be performed in time $N^{O(d)}$ with, for instance, the ellipsoid algorithm.

To give a concrete formulation of this SDP, we now describe $\fE_d^N$ in terms of the \emph{pseudomoment} interpretation of SOS optimization (as derived from the general formulation of SOS in, e.g., \cite{Laurent-SOS}).
Below we adopt the useful notations of $\binom{[N]}{k}$ and $\binom{[N]}{\leq k}$ for the sets of subsets of $[N]$ having size $k$ and size at most $k$ (including the empty set), respectively.
We also use the standard notation $S \triangle T \colonequals (S \setminus T) \cup (T \setminus S)$ for the symmetric difference of the sets $S$ and $T$.

\begin{definition}
    \label{def:pm}
    $\fE_d^N \subset \RR^{N \times N}_{\sym}$ is the set of matrices $\bM$ such that there exists $\bZ \in \RR^{\binom{[N]}{\leq d / 2} \times \binom{[N]}{\leq d / 2}}$ having
    \begin{equation}
        Z_{\{i\}\{j\}} = M_{ij} \text{  for all } i, j \in [N]
    \end{equation}
    and satisfying the following properties:
    \begin{enumerate}
    \item $\bZ \succeq \bm 0$.
    \item $Z_{ST}$ only depends on $S \triangle T$.
    \item $Z_{ST} = 1$ whenever $S \triangle T = \emptyset$.
    \end{enumerate}
    In this case, we say $\bZ$ is a \emph{degree $d$ pseudomoment matrix for the constraint polynomials $\{x_i^2 - 1: i \in [N]\} \subset \RR[x_1, \dots, x_N]$}, which \emph{extends} $\bM$.\footnote{Usually, a degree $d$ pseudomoment matrix must be indexed by all monomials in $x_1, \dots, x_N$ of degree at most $d / 2$; however, in our case, the constraint $x_i^2 = 1$ ensures that the pseudomoments of multilinear monomials fully determine the pseudomoment matrix.}
\end{definition}
\noindent
For the sake of brevity, we will simply refer to such $\bZ$ as a \emph{degree $d$ pseudomoment matrix}, since we only study optimization over $\{ \pm 1\}^N$.
\begin{definition}
    The \emph{degree $d$ sum-of-squares relaxation} of $\mathsf{M}(\bW)$ is
    \begin{equation}
        \SOS_d(\bW) \colonequals \frac{1}{N}\max_{\bM \in \fE_d^N} \la \bW, \bM \ra.
    \end{equation}
\end{definition}
\noindent
In addition to Montanari's general question on certifying bounds on $\mathsf{M}(\bW)$ mentioned before, Jain, Risteski, and Koehler have independently posed the more specific question of determining the asymptotic value of $\SOS_d(\bW)$ when $\bW \sim \GOE(N)$ in \cite{JRK-SOS}.

One further simplification of the above setup will also be useful to take into account.
\begin{definition}
    We call $\bZ \in \RR^{\binom{[N]}{\leq d / 2} \times \binom{[N]}{\leq d / 2}}$ a \emph{reduced degree $d$ pseudomoment matrix} if it satisfies all of the conditions of Definition~\ref{def:pm}, as well as the following additional condition:
    \begin{enumerate}
    \item[4.] $Z_{ST} = 0$ whenever $|S \triangle T|$ is odd.
    \end{enumerate}
\end{definition}
As we show below, because of the invariance of both the constraints $x_i^2 - 1 = 0$ and the objective function $\bx^{\top}\bW\bx$ under the map $\bx \mapsto -\bx$, $\fE_d^N$ may be equivalently defined in terms of reduced pseudomoment matrices.
\begin{proposition}
    \label{prop:pseudomoment-mx-reduction}
    If there exists a degree $d$ pseudomoment matrix extending $\bM$, then there exists a reduced degree $d$ pseudomoment matrix extending $\bM$.
\end{proposition}
\begin{proof}
    It suffices to show that if $\bZ^{(0)}$ is a degree $d$ pseudomoment matrix, then the matrix $\bZ^{(1)}$ of the same shape with entries $Z^{(1)}_{ST} \colonequals (-1)^{|S \triangle T|} Z^{(0)}_{ST}$ is also a degree $d$ pseudomoment matrix, since then $\frac{1}{2}\bZ^{(0)} + \frac{1}{2}\bZ^{(1)}$ is a reduced degree~$d$ pseudomoment matrix whose minor indexed by $\binom{[N]}{1}$ equals that of $\bZ^{(0)}$.

    Clearly all of the linear constraints on $\bZ^{(1)}$ are satisfied, so it suffices to check positive semidefiniteness.
    For this, if $\bx^{(1)} \in \RR^{\binom{[N]}{\leq d / 2}}$, let us define $\bx^{(0)} \in \RR^{\binom{[N]}{\leq d / 2}}$ by $x^{(0)}_S = (-1)^{|S|}x^{(0)}_S$.
    Then, since $(-1)^{|S| + |T|} = (-1)^{|S \triangle T|}$, we have $\bx^{(1)^{\top}} \bZ^{(1)} \bx^{(1)} = \bx^{(0)^{\top}} \bZ^{(0)} \bx^{(0)} \geq 0$, completing the proof.
\end{proof}

We will only study degree 2 and degree 4 pseudomoment matrices in detail, so we give more concrete versions of the above conditions for those cases.
\begin{proposition}
    \label{prop:deg-2-pm-def}
    $\fE_2^N = \{\bM \in \RR^{N \times N}_{\sym} : \bM \succeq \bm 0, M_{ii} = 1 \text{ for all } i \in [N]\}.$
\end{proposition}
\begin{proposition}
    \label{prop:deg-4-pm-def}
    Let $\bZ \in \RR^{(N(N + 1) / 2 + 1) \times (N(N + 1) / 2 + 1)}_{\sym}$, with the row and column indices of $\bZ$ identified with $\binom{[N]}{\leq 2}$, ordered first by size and then in lexicographical order.\footnote{For instance, $\binom{[3]}{\leq 2}$ is ordered as $\emptyset, \{1\}, \{2\}, \{3\}, \{1, 2\}, \{1, 3\}, \{2, 3\}$.}
    Then, $\bZ$ is a reduced degree 4 pseudomoment matrix if and only if the following conditions hold:
    \begin{enumerate}
    \item $\bZ \succeq \bm 0$.
    \item $Z_{\emptyset\emptyset} = Z_{\{i\}\{i\}} = Z_{\{i,j\}\{i,j\}} = 1$ for all distinct $i, j \in [N]$.
    \item $Z_{\{i\}\emptyset} = Z_{\{i,j\}\{j\}} = 0$ for all distinct $i, j \in [N]$.
    \item $Z_{\{i, j\}\{i, k\}} = Z_{\{j, k\}\emptyset} = Z_{\{j\}\{k\}}$ for all distinct $i, j, k \in [N]$.
    \item $Z_{\{i,j\}\{k,\ell\}}$ is invariant under permutations of the indices $i, j, k, \ell$.
    \end{enumerate}
\end{proposition}

\subsection{Prior work on sum-of-squares lower bounds}
\label{sec:montanari-sen}

\paragraph{The Montanari-Sen degree 2 lower bound}
The only previous result on SOS relaxations of $\mathsf{M}(\bW)$ under $\bW \sim \GOE(N)$ that we are aware of is the following result of Montanari and Sen \cite{MS-deg2}, which establishes hardness of certification for degree 2 SOS.
\begin{theorem}[Theorem 5(a) of \cite{MS-deg2}]
    \label{thm:ms}
    Let $\epsilon > 0$.
    Then,
    \begin{equation}
        \lim_{N \to \infty} \PP\left[\SOS_2(\bW) \geq 2 - \epsilon \right] = 1.
    \end{equation}
\end{theorem}

\noindent
The mechanics of this result will be crucial for proving ours, so let us also review how to construct
feasible points achieving this bound. First, fix a parameter $\delta \in (0, 1)$, and set $r = r(N) = \delta N$.\footnote{Following \cite{MS-deg2}, we assume for the sake of simplicity that $\delta N$ is an integer; recovering the same results for $r = \lfloor \delta N \rfloor$ is tedious but straightforward.}
    Then, given $\bW \sim \GOE(N)$, let $V \subset \RR^N$ be the subspace spanned by the $r(N)$ eigenvectors of $\bW$ having the largest eigenvalues, and let $\bP$ be the orthogonal projector to $V$.
    Let $\bD \in \RR^{N \times N}$ be the diagonal matrix with entries $D_{ii} = P_{ii}$, and define
    \begin{equation}
        \bM^{(\delta)} = \bM^{(\delta)}(\bW) \colonequals \bD^{-1/2} \bP \bD^{-1/2}.
    \end{equation}
    (Note that $D_{ii} = 0$ if and only if $\bm e_i \in V^\perp$, which almost surely does not occur.)
    Then, $\bM^{(\delta)}(\bW) \in \fE_2^N$ for any $\delta \in (0, 1)$ almost surely, and for each $\epsilon > 0$ there exists $\delta \in (0, 1)$ such that
    \begin{equation}
        \lim_{N \to \infty}\PP\left[\frac{1}{N}\la \bM^{(\delta)}(\bW), \bW \ra \geq 2 - \epsilon\right] = 1.
    \end{equation}
We call $\bM^{(\delta)}$ the \emph{Montanari-Sen witness}.
This construction will be the basis of ours; indeed, we will show that only a small correction is necessary to make this witness feasible for the degree 4 SOS relaxation as well.

\paragraph{Simultaneous work on degree 4 lower bounds}
After preparing an earlier version of this manuscript, we learned of the concurrent work \cite{MRX-2019-Lifting2To4}, the fruit of a parallel research effort in which the same result (Corollary~\ref{cor:sos-value}, stated in the following section) as ours is proved.
Their construction uses the \emph{pseudocalibration} sum-of-squares heuristic introduced in \cite{BHKKMP-pc}, while our construction is based on more geometric and problem-specific considerations of constraints on the pseudomoment extension matrices described above.
We remark that one important advantage of their work is that its analysis is general enough to apply immediately to problems besides the SK model, such as MaxCut on sparse random graphs.

On the other hand, the output of pseudocalibration in this situation appears to be quite complicated: while the degree 4 pseudomoments constructed in \cite{MRX-2019-Lifting2To4}, like ours, are low-degree polynomials in the entries of $\bM$, they are not stated explicitly, and are sums of many more ``graphical'' terms, simple polynomials in the entries of $\bM$ described by graphs labelled by matrix indices.
For instance, while the formula we will derive for our degree 4 pseudomoments includes just two types of such polynomials (see our Section~\ref{sec:construction}), the analysis of the pseudomoments of \cite{MRX-2019-Lifting2To4} involves at least 20 types (see their Section 3).
Thus, though it is possible that the degree 4 pseudomoment matrices constructed in our work are close to those constructed in \cite{MRX-2019-Lifting2To4} (in spectral norm, for example), our approach gives a more explicit and condensed description of a suitable pseudomoment construction.

Furthermore, while both this paper and \cite{MRX-2019-Lifting2To4} use the idea of an ``approximate Cholesky decomposition'' of the pseudomoment matrix (also present in earlier works such as \cite{BHKKMP-pc}) for the proofs of positive semidefiniteness, we give in Sections~\ref{sec:heuristic-gauss} and \ref{sec:conj-extension} an intuitive probabilistic meaning to the Cholesky factors arising in this argument.
To the best of our knowledge, this is a novel interpretation, which we believe may shed new light on the structure of pseudomoment matrices and simplify some of the ad hoc technicalities arising in the proofs of their positive semidefiniteness.

\section{Results}

Our main result establishes hardness of certification for degree 4 SOS, by showing that a minor variant of the Montanari-Sen witness admits a degree 4 extension.
\begin{theorem}
    \label{thm:sos4-witness}
    Let $\alpha, \delta \in (0, 1)$, and let $\bM^{(\delta)}$ be as in Theorem~\ref{thm:ms}.
    Define
    \begin{equation}
        \bM^{(\alpha, \delta)}(\bW) \colonequals (1 - \alpha)\bM^{(\delta)}(\bW) + \alpha \bm I_N.
    \end{equation}
    Then,
    \begin{equation}
        \lim_{N \to \infty}\PP\left[\bM^{(\alpha, \delta)}(\bW) \in \fE_4^N \right] = 1.
    \end{equation}
\end{theorem}
\noindent
The theorem says that an arbitrarily small adjustment in the direction the identity matrix (the barycenter of the vertices of the cut polytope) suffices to make the degree 2 Montanari-Sen primal witness admit a degree 4 extension with high probability.\footnote{The specific choice of the identity matrix is probably not essential, but is convenient because, as we will see, the degree 4 extension of the identity matrix has an especially simple spectral structure.}
In Conjecture~\ref{conj:all-degrees-elliptope}, we will also propose that the same membership holds with high probability for any SOS relaxation $\fE_d^N$ with constant degree $d$ as $N \to \infty$.

Since the degree 2 part of our construction is a simple modification of the Montanari-Sen witness, it is straightforward to apply Theorem~\ref{thm:ms} and obtain a lower bound for the degree 4 SOS objective.
\begin{corollary}
    \label{cor:sos-value}
    Let $\epsilon > 0$.
    Then,
    \begin{equation}
        \lim_{N \to \infty} \PP\left[\SOS_4(\bW) \geq 2 - \epsilon \right] = 1.
    \end{equation}
\end{corollary}
\begin{proof}
    By Theorem~\ref{thm:ms}, take $\delta \in (0, 1)$ such that $\frac{1}{N}\la \bM^{(\delta)}(\bW), \bW \ra \geq 2 - \epsilon / 3$ with high probability.
    Let $\alpha \in (0, 1)$ be small that $(1 - \alpha)(2 - \epsilon / 3) \geq 2 - 2 \epsilon / 3$.
    By Theorem~\ref{thm:sos4-witness}, with high probability,
    \begin{equation}
        \SOS_4(\bW) \geq \frac{1}{N}\la \bM^{(\alpha, \delta)}(\bW), \bW \ra \geq 2 - \frac{2\epsilon}{3} + \frac{1}{N}\alpha \Tr(\bm W).
    \end{equation}
    The random variable $\Tr(\bm W)$ has law $\sN(0, 2)$, so with high probability the last term is smaller than $\epsilon / 3$, and the result follows.
\end{proof}

\paragraph{Organization}
The remainder of the paper gives the proof of Theorem~\ref{thm:sos4-witness}, and conjectures an extension of the same proof technique to higher degree SOS relaxations.
In Section~\ref{sec:preliminaries} we review some preliminary facts and notations.
In Section~\ref{sec:construction} we give the motivation and precise statement of our construction of a matrix that is with high probability a degree 4 pseudomoment matrix extending $\bM^{(\alpha, \delta)}(\bW)$.
In Section~\ref{sec:conj-extension}, we present a natural conjecture for how this construction can be extended to higher degrees.
Finally, in Sections~\ref{sec:proof-outline}, \ref{sec:main-term}, and \ref{sec:corr-term}, we prove that our construction from Section~\ref{sec:construction} indeed furnishes a valid degree 4 pseudomoment matrix with high probability.

\section{Preliminaries}
\label{sec:preliminaries}

\subsection{Glossary: variables, parameters, and constants}

For reference, we summarize various symbols that will appear throughout.
The following are the general scalar parameters involved.
\begin{itemize}
\item $N$ is a parameter indicating the size of $\bW \sim \GOE(N)$ in the problem statement.
\item $\delta \in (0, 1)$ is a fixed parameter not depending on $N$.
\item $r = r(N) \colonequals \delta N$, which we assume for the sake of simplicity is an integer.
    This gives the rank of the Montanari-Sen witness that we extend (before modifications that make it have full rank)
\item $\alpha \in (0, 1)$ is another fixed parameter not depending on $N$ or $\delta$.
\item $K$ is a constant appearing in concentration inequalities, giving polynomial rates of decay of probabilities of the form $N^{-K}$. All of the concentration inequalities where $K$ appears hold with any choice of $K > 0$, but the other constants appearing in those results depend on both $K$ and $\delta$. Thus a typical inequality will take the form
\begin{equation}
    \text{``} \PP\left[\text{quantity} \leq O_{\delta, K}(N)\right] \geq 1 - O_{\delta, K}(N^{-K}). \text{''}
\end{equation}
One may think of any concrete choice, e.g.\ $K = 100$, throughout.
\end{itemize}
The following are vectors and matrices associated to the Montanari-Sen construction applied to a specific instance $\bW$.
\begin{itemize}
    \item $\bV \in \RR^{r \times N}$ is the matrix having the top $r$ (unit norm) eigenvectors of $\bW$ as its rows.
    \item $\bv_1, \dots, \bv_N \in \RR^r$ are the columns of $\bV$ (\emph{not} eigenvectors of $\bW$).
    \item $\bP = \bV^\top \bV$ is the orthogonal projector to the top $r$-dimensional eigenspace of $\bW$, and also the Gram matrix of the $\bv_i$.
    \item $\what{\bv}_i = \bv_i / \|\bv_i\|_2$.
    \item $\bD \in \RR^{N \times N}$ is a diagonal matrix with $D_{ii} = \|\bv_i\|_2^2 = P_{ii}$.
    \item $\what{\bV} = \bV \bD^{-1/2} \in \RR^{r \times N}$ is the matrix having the $\what{\bv}_i$ as its columns.
    \item $\bM = \bM^{(\delta)} = \what{\bV}^\top \what{\bV} = \bD^{-1/2} \bP \bD^{-1/2}$ is the Montanari-Sen witness, and also the Gram matrix of the $\what{\bv}_i$. For the sake of brevity, we will often drop the $(\delta)$ superscript, as $\delta$ will be a constant carried throughout.
\end{itemize}

\subsection{Other notation}

\paragraph{Linear algebra} The identity matrix in dimension $n$ is denoted $\bm I_n$, and the all-ones vector is denoted $\one_n$.
The Frobenius or entrywise inner product of matrices having the same shape is denoted $\la \bA, \bB \ra = \Tr(\bA^\top \bB) = \sum_{i, j} A_{ij}B_{ij}$.
The Hadamard or entrywise product of matrices having the same shape is denoted $\bA \circ \bB$ and has entries $(\bA \circ \bB)_{ij} = A_{ij}B_{ij}$.
Hadamard powers of a matrix are denoted $\bA^{\circ k}$.
The matrix operator norm is denoted $\|\bA\|_{\op} = \max_{\|\bx\|_2 = \|\by\|_2 = 1} \bx^\top \bA \by$.
The vectorized supremum norm is denoted $\|\bA\|_{\ell^\infty} = \max_{i, j} |A_{ij}|$.
The matrix Frobenius norm is denoted $\|\bA\|_F = \la \bA, \bA \ra^{1/2}$.
The group of $N \times N$ orthogonal matrices is denoted $\sO(N)$.
The Stiefel manifold of $r \times N$ matrices with orthonormal rows is denoted $\mathsf{Stief}(N, r)$.

\paragraph{Vectorization} We describe several ways of vectorizing symmetric matrices. For $\bA \in \RR^{n \times n}_{\sym}$, we let $\diag(\bA) \in \RR^{n}$ and $\offdiag(\bA) \in \RR^{n(n - 1) / 2}$ be the vectorized diagonal and strict upper triangle of $\bA$, respectively, with index sets ordered lexicographically.
These two vectors determine a symmetric matrix completely.
It will also be useful to define an isometry between $\RR^{n \times n}_{\sym}$ endowed with the Frobenius inner product and $\RR^{n(n + 1) / 2}$ endowed with the ordinary Euclidean inner product.
This is given by
\begin{equation}
    \isovec(\bA) \colonequals \left[\begin{array}{c} \diag(\bA) \\ \sqrt{2} \cdot \offdiag(\bA)\end{array}\right],
\end{equation}
which indeed satisfies $\la \isovec(\bA), \isovec(\bB) \ra = \la \bA, \bB \ra$.
We denote $\one_{\diag} \colonequals \isovec(\bm I_n)$ when the dimension $n$ is clear from context, since this is the indicator vector of the diagonal matrix indices.

\paragraph{Probability} We denote the relation of two random variables $X, Y$ having the same law by $X \eqd Y$.
The Haar measure on the orthogonal group or Stiefel manifold is denoted $\mathsf{Haar}(\sO(N))$ or $\mathsf{Haar}(\mathsf{Stief}(N, r))$, respectively, and refers to Haar measure with respect to the action of $\sO(N)$ by right multiplication in the latter case.
We write $\mathsf{Unif}(\SS^{N - 1})$ for the Haar measure on vectors of the unit sphere under the action of $\sO(N)$ by left multiplication, equivalent to $\mathsf{Haar}(\mathsf{Stief}(N, 1))$ (up to transposition).

\subsection{Basic properties of the Montanari-Sen witness}

It will be useful to establish some preliminary bounds on and distributional properties of the Montanari-Sen construction described in Theorem~\ref{thm:ms}.

First, we use an elegant geometric argument mentioned in \cite{MS-deg2} (but which seems to be well-known folklore) to obtain bounds on the entries of $\bP$.
\begin{proposition}
    \label{prop:proj-entries}
    For all $K > 0$,
    \begin{equation} \PP\left[\max_{i, j \in [N]} \left\{\begin{array}{lll} \left|P_{ii} - \delta\right| & \text{if} & i = j \\ \left|P_{ij}\right| & \text{if} & i \neq j\end{array}\right\} \leq O_{\delta, K}\left(\sqrt{\frac{\log N}{N}}\right) \right] \geq 1 - O_{\delta, K}(N^{-K}). \end{equation}
\end{proposition}
\begin{proof}
 For the diagonal entries, $P_{ii} = \|\bP\be_i\|_2^2 \eqd \|\bP_0\bv\|_2^2$ for $\bv \sim \mathsf{Unif}(\SS^{N - 1})$ and $\bP_0$ the orthogonal projector onto $\mathsf{span}(\be_1, \dots, \be_r)$.
    Letting $\bg \sim \sN(\bm 0, \bm I_N)$, we have $\bv \eqd \bg / \|\bg \|_2$.
    Writing $\bg^{(r)}$ and $\bg^{(N - r)}$ for the first $r$ and last $N - r$ coordinates of $\bg$ respectively, we then have
    \begin{equation} P_{ii} \eqd \frac{\|\bg^{(r)}\|_2^2}{\|\bg^{(r)}\|_2^2 + \|\bg^{(N - r)}\|_2^2}, \end{equation}
    where $\|\bg^{(r)}\|_2^2$ and $\|\bg^{(N - r)}\|_2^2$ are distributed as independent $\chi^2$ random variables with $r$ and $N - r$ degrees of freedom respectively.
    The result then follows from the concentration inequalities of \cite{LM-chi-squared} for $\chi^2$ random variables.

    For the off-diagonal entries, likewise $P_{ij} = \la \bP \be_i, \bP \be_j \ra \eqd \la \bP_0\bv, \bP_0 \bw \ra$, where now $(\bv, \bw)$ are a Haar-distributed two-dimensional orthonormal frame.
    If we draw $\bg, \bh \sim \sN(\bm 0, \bm I_N)$ independently, then by performing two steps of Gram-Schmidt orthonormalization,
    \begin{equation} (\bv, \bw) \eqd \left(\frac{\bg}{\|\bg\|_2}, \frac{\bh - \frac{\la \bg, \bh \ra}{\|\bg\|_2^2} \bg}{\left\|\bh - \frac{\la \bg, \bh \ra}{\|\bg\|_2^2} \bg\right\|_2}\right). \end{equation}
    Thus computing as before we find
    \begin{equation} P_{ij} \eqd \frac{\la \bg^{(r)}, \bh^{(r)} \ra - \la \bg, \bh \ra \frac{\|\bg^{(r)}\|_2^2}{\|\bg^{(r)}\|_2^2 + \|\bg^{(N - r)}\|_2^2}}{\|\bg\|_2\left\|\bh - \frac{\la \bg, \bh \ra}{\|\bg\|_2^2} \bg\right\|_2}. \end{equation}
    To control these quantities, we first bound
    \begin{equation}
        \left|\frac{\la \bg^{(r)}, \bh^{(r)} \ra - \la \bg, \bh \ra \frac{\|\bg^{(r)}\|_2^2}{\|\bg^{(r)}\|_2^2 + \|\bg^{(N - r)}\|_2^2}}{\|\bg\|_2\left\|\bh - \frac{\la \bg, \bh \ra}{\|\bg\|_2^2} \bg\right\|_2}\right| \leq \frac{|\la \bg^{(r)}, \bh^{(r)} \ra| + |\la \bg, \bh \ra|}{\|\bg\|_2(\|\bh\|_2 - \frac{|\la \bg, \bh \ra|}{\|\bg\|_2})}.
    \end{equation}
    In this expression, every $\ell^2$ norm may be controlled by the concentration inequalities of \cite{LM-chi-squared} as before, and every inner product may be controlled by observing that, when $\bg$ and $\bh$ are independent standard gaussian vectors, then $\la \bg, \bh \ra \eqd \|\bg\|_2 h_1$, thus reducing the task of controlling the inner product to controlling an $\ell^2$ norm.
    Applying these bounds then gives the result.
\end{proof}

The following related results for $\bM$ follow directly.
\begin{corollary}
    \label{cor:M-entries}
    For all $K > 0$,
    \begin{equation} \PP\left[\max_{i, j \in [N]} \left\{\begin{array}{lll}\left|M_{ii} - 1\right| & \text{if} & i = j \\ \left|M_{ij}\right| & \text{if} & i \neq j\end{array}\right\} \leq O_{\delta, K}\left(\sqrt{\frac{\log N}{N}}\right) \right] \geq 1 - O_{\delta, K}(N^{-K}). \end{equation}
\end{corollary}
\begin{proof}
    We have $M_{ij} = P_{ij} / \sqrt{P_{ii}P_{jj}}$, so the result follows from combining Proposition~\ref{prop:proj-entries} applied to several entries.
\end{proof}

\begin{corollary}
    \label{cor:M-norm}
    For all $K > 0$,
    \begin{equation} \PP\left[\|\bM\|_{\op} \leq \delta^{-1} + O_{\delta, K}\left(\sqrt{\frac{\log N}{N}}\right)\right] \geq 1 - O_{\delta, K}(N^{-K}). \end{equation}
\end{corollary}
\begin{proof}
    $\bM = \bD^{-1/2}\bP\bD^{-1/2}$, so $\|\bM\|_{\op} \leq \|\bD^{-1/2}\|_{\op}^2 = (\min_{i \in [N]} P_{ii})^{-1}$, and the result then follows by Proposition~\ref{prop:proj-entries}.
\end{proof}

\subsection{Moments of $\Haar(\sO(N))$}

The following gives the low-degree moments of Haar-distributed orthogonal matrices.

\begin{proposition}[Lemma 9 of \cite{CM-haar}]
    \label{prop:haar-moments}
    Let $\bQ \sim \Haar(\sO(N))$.
    The moment $\EE \prod_{k = 1}^d Q_{i_kj_k}$ is zero if any index occurs an odd number of times among either the $i_k$ or $j_k$.
    The non-zero degree 2 and 4 moments are given by
    \begin{align}
      \EE Q_{11}^2 &= \frac{1}{N}, \\
      \EE Q_{11}^4 &= \frac{3}{N(N + 2)}, \\
      \EE Q_{11}^2Q_{12}^2 &= \frac{1}{N(N + 2)}, \\
      \EE Q_{11}^2Q_{22}^2 &= \frac{N + 1}{(N - 1)N(N + 2)}, \\
      \EE Q_{11}Q_{12}Q_{21}Q_{22} &= -\frac{1}{(N - 1)N(N + 2)}.
    \end{align}
\end{proposition}

\section{Degree 4 pseudomoment construction}
\label{sec:construction}

In this section we outline the main idea in the proof of Theorem~\ref{thm:sos4-witness}, detailing the construction of a suitable degree 4 pseudomoment matrix and reducing Theorem~\ref{thm:sos4-witness} to verifying the positive semidefiniteness of this matrix.

First, we give two heuristic lines of reasoning supporting the construction we propose.
The first derives a formula for the degree 4 pseudomoment extension of a highly structured collection of degree 2 pseudomoments and assumes that this formula may be transferred verbatim to the random case.
The second recovers the same formula after applying some simplifying heuristics to a probabilistic argument.

As a preliminary, let us remark that, per Proposition~\ref{prop:pseudomoment-mx-reduction}, we may restrict our attention to reduced degree 4 pseudomoment matrices.
If $\bZ$ is a reduced degree 4 pseudomoment matrix extending $\bM$, and we divide $\bZ$ into blocks indexed by $\binom{[N]}{0}$, $\binom{[N]}{1}$, and $\binom{[N]}{2}$, where $\bZ^{[i, j]}$ is the block with rows indexed by $\binom{[N]}{i}$ and columns indexed by $\binom{[N]}{j}$, then only the block $\bZ^{[2, 2]}$ is not determined by the properties of being reduced and extending $\bM$:
\begin{equation}
    \label{eq:deg-4-pm-blocks}
    \bZ \equalscolon \left[\arraycolsep 4pt \begin{array}{ccc} \bZ^{[0, 0]} & \bZ^{[0, 1]} & \bZ^{[0, 2]} \\ \bZ^{[1, 0]} & \bZ^{[1, 1]} & \bZ^{[1, 2]} \\ \bZ^{[2, 0]} & \bZ^{[2, 1]} & \bZ^{[2, 2]} \end{array}\right] = \left[\arraycolsep 4pt \begin{array}{ccc} 1 & \bm 0 & \offdiag(\bM)^{\top} \\ \bm 0 & \bM & \bm 0 \\ \offdiag(\bM) & \bm 0 & \bZ^{[2, 2]} \end{array}\right].
\end{equation}
Thus in the sequel we will be justified in saying that an extension $\bZ$ of $\bM$ is specified just by a particular choice of $\bZ^{[2, 2]}$.

\subsection{Heuristic 1: evidence from equiangular tight frames}

\label{sec:heuristic-etf}

In this section, we review a result from the works \cite{BK-ETF,BK-sampta} of the authors', which derived an explicit description of degree 4 extensions of degree 2 pseudomoment matrices for some very structured special cases, that resemble the Montanari-Sen witness in that their degree 2 pseudomoment matrices are constant multiples of orthogonal projectors.

The extra structure that allows this description of the degree 4 pseudomoments to be derived in closed form is described by the following notions from finite frame theory.
\begin{definition}
    A collection of vectors $\what{\bv}_1, \dots, \what{\bv}_N \in \RR^r$ forms a \emph{unit norm tight frame (UNTF)} if the following conditions hold.
    \begin{enumerate}
        \item (Unit Norm) $\|\what{\bv}_i\|_2 = 1$ for all $i \in [N]$.
        \item (Tight Frame) $\sum_{i = 1}^N \what{\bv}_i\what{\bv}_i^\top = \frac{N}{r}\bm I_r$.
    \end{enumerate}
    They moreover form an \emph{equiangular tight frame (ETF)} if the following additional condition holds.
    \begin{enumerate}
    \item[3.] (Equiangular) There is $\mu \in [0, 1]$ such that $|\la \what{\bv}_i, \what{\bv}_j \ra | = \mu$ whenever $i \neq j$.
    \end{enumerate}
\end{definition}
\noindent
ETFs are rare and combinatorially structured objects \cite{Tropp-ETF,Casazza-ETF,FM-ETF}, which do not seem \emph{a priori} related to the SK problem and the Montanari-Sen witness.
However, it turns out that studying the degree 4 extensions of ETFs gives useful insight into the correct construction of pseudomoments even in the random case.

In general, we showed in \cite{BK-ETF} that degree 4 extensions are related to the following notion from convex geometry.
\begin{definition}
    Let $K \subseteq \RR^d$ be a closed convex set.
    For $\bM \in K$, the \emph{perturbation of $\bM$ in $K$} is the linear subspace
    \begin{equation}
        \mathsf{pert}_K(\bM) \colonequals \left\{\bA \in \RR^d: \bM \pm t\bA \in K \text{ for all } t > 0 \text{ sufficiently small }\right\}.
    \end{equation}
    Equivalently, if $H$ is the affine hull of the minimal face of $K$ that $\bM$ belongs to, then $\mathsf{pert}_K(\bM) = H - \bM$.
\end{definition}
\noindent
The relevant case for degree 4 extensions is $K = \fE_2^N$.
In this case, the following theorem characterizes the perturbation subspace.
\begin{proposition}[Theorem 1(a) of \cite{li:94}]
    \label{prop:li-tam}
    Let $\bM \in \fE_2^N$ have $\rank(\bM) = r$ and $\bM = \what{\bV}^\top \what{\bV}$ for $\what{\bV} \in \RR^{r \times N}$ with unit vector columns $\what{\bv}_1, \dots, \what{\bv}_N$.
    Then,
    \begin{align}
      \mathsf{pert}_{\fE_2^N}(\bM)
      &= \left\{ \what{\bV}^\top \bS \what{\bV}: \bS \in \RR^{r \times r}_{\sym} \right\} \cap \{ \bA \in \RR^{N \times N}: \diag(\bA) = \bm 0 \} \nonumber \\
      &= \left\{ \what{\bV}^\top \bS \what{\bV}: \bS \in \RR^{r \times r}_{\sym}, \what{\bv}_i^\top \bS \what{\bv}_i = 0 \text{ for } i \in [N] \right\} \nonumber \\
      &= \what{\bV}^\top \left( \mathsf{span}\left(\{\what{\bv}_1\what{\bv}_1^\top, \dots, \what{\bv}_N\what{\bv}_N^\top\}\right)^\perp\right) \what{\bV}.
    \end{align}
\end{proposition}

For ETFs, a degree 4 extension, when any exists, is given explicitly in both spectral and entrywise terms as follows.
\begin{theorem}[Theorem 2.19 of \cite{BK-ETF}]
    \label{thm:etf-extension}
    Let $\bM \in \fE_2^N$ be the Gram matrix of an ETF of $N$ vectors in $\RR^r$.
    Then, $\bM \in \fE_4^N$ if and only if $N < \frac{r(r + 1)}{2}$.
    In this case, a degree 4 extension $\bZ \in \RR^{(N(N + 1) / 2 + 1) \times (N(N + 1) / 2 + 1)}_{\sym}$ is given by
    \begin{equation}
        \bZ^{[2,2]} = \offdiag(\bM)\offdiag(\bM)^\top + \frac{N^2(1 - \frac{1}{r})}{r(r + 1) - 2N} \bP_{\offdiag(\mathsf{pert}_{\fE_2^N}(\bM))}
    \end{equation}
    and the block structure given in \eqref{eq:deg-4-pm-blocks}.
    The entries of this matrix are given by
    \begin{align}
        Z_{\{i,j\}\{k,\ell\}}^{[2,2]}
        &= \frac{\frac{r(r - 1)}{2}}{\frac{r(r + 1)}{2} - N}(M_{ij}M_{k\ell} + M_{ik}M_{j\ell} + M_{i\ell}M_{jk}) \nonumber \\
        &\hspace{1cm}- \frac{r^2(1 - \frac{1}{N})}{\frac{r(r + 1)}{2} - N}\sum_{m = 1}^N M_{im}M_{jm}M_{km}M_{\ell m}.
    \end{align}
\end{theorem}
\noindent
(The spectral description will be of interest later, to draw a connection to the second heuristic presented in the following section.)

The Montanari-Sen witness $\bM$ is close to a constant multiple of a projection matrix, since by Proposition~\ref{prop:proj-entries} all entries of the normalizing diagonal matrix $\bD$ are close to $\delta$, so $\bM$ is a ``near-UNTF Gram matrix.''
Also, the off-diagonal entries of $\bM$ are inner products of random unit vectors $\what{\bv}_i$, which it is reasonable to think are weakly dependent, whereby $\bM$ should moreover behave like ``an ETF in expectation.''

Thus to guess a degree 4 extension for the Montanari-Sen witness $\bM$, we may be justified in simply trying to apply the combinatorial ETF construction directly.
Since we take $r = \delta N$ with $r, N \to \infty$, we may also simplify the leading coefficients to their asymptotic values, which gives the prediction
\begin{equation}
    \label{eq:pm-pred-1}
    \text{``}Z_{\{i,j\}\{k,\ell\}}^{[2,2]} = M_{ij}M_{k\ell} + M_{ik}M_{j\ell} + M_{i\ell}M_{jk} - 2 \sum_{m = 1}^N M_{im}M_{jm}M_{km}M_{\ell m}.\text{''}
\end{equation}

\subsection{Heuristic 2: conditional covariance of gaussian matrices}

\label{sec:heuristic-gauss}

We next show another, perhaps more principled argument through which we arrive at the same prediction of degree 4 pseudomoments.
Let us suppose that $\bM$ is exactly the Gram matrix of a UNTF, i.e.\ $\bM = \delta^{-1}\bP = \delta^{-1}\bV^\top \bV$ and $\diag(\bP) = \delta \one_N$.
As mentioned above, this is approximately the case for the Montanari-Sen witness, but our heuristic derivation is much simplified if the correction by the diagonal matrix $\bD$ to form the actual Montanari-Sen witness may be reduced to a constant scaling.

For the computations to come it is more convenient to view $\bZ^{[2, 2]}$ as being embedded in a larger pseudomoment matrix, indexed by pairs $(i, j) \in [N]^2$.
We denote this matrix by $\bZ^{\mathsf{mult}[2, 2]}$, and, following the pseudomoment framework discussed before, set $Z^{\mathsf{mult}[2, 2]}_{(i, j)(k, \ell)}$ to be the ``pseudoexpectation'' of the monomial $x_ix_jx_kx_{\ell}$.
In terms of $\bZ^{[2, 2]}$ and $\bM$, these entries are
\begin{equation}
    Z^{\mathsf{mult}[2, 2]}_{(i, j)(k, \ell)} = \left\{\begin{array}{ll} Z^{[2, 2]}_{\{i, j\}\{k, \ell\}} & \text{if } i \neq j \text{ and } k \neq \ell, \\
                                                         M_{ij} & \text{if } k = \ell, \\
                                                         M_{k\ell} & \text{if } i = j.\end{array}\right.
\end{equation}
We note that if $\bZ$ is positive semidefinite then so is $\bZ^{\mathsf{mult}[2, 2]}$, since, up to removing repeated rows and columns, the latter is equal to the principal minor of the former indexed by $\{\emptyset\} \cup \binom{[N]}{2}$.

Now, we view $\bZ^{\mathsf{mult}[2,2]}$, a positive semidefinite $N^2 \times N^2$ matrix, as the degree 2 moment matrix of the entries of a gaussian random matrix: we suppose there exists some $\bA \in \RR^{N \times N}$ with random jointly gaussian entries such that
\begin{equation}
    Z_{(i,j)(k,\ell)}^{\mathsf{mult}[2,2]} = \EE [A_{ij}A_{k\ell}].
\end{equation}
We then design $\bA$ so that $\bZ^{\mathsf{mult}[2, 2]}$ automatically satisfies some of the necessary constraints, and hope that the remaining constraints will be approximately satisfied as well.

We begin with a matrix $\bA^{(0)}$ having a canonical gaussian distribution for symmetric matrices, the GOE, suitably rescaled to allow us a normalizing degree of freedom later: $A_{ii}^{(0)} \sim \sN(0, 2\sigma^2)$ and $A_{ij}^{(0)} = A_{ji}^{(0)} \sim \sN(0, \sigma^2)$.
Taking $\bA$ to be symmetric already ensures some of the symmetry conditions that $\bZ^{\mathsf{mult}[2,2]}$ must satisfy.
Next, we take $\bA$ to have the distribution of $\bA^{(0)}$, conditional on the following two properties:
\begin{enumerate}
    \item $(\bm I_N - \bP)\bA = 0$.
    \item $A_{ii} = 1$ for all $i \in [N]$.
\end{enumerate}
Property 1 ensures that any symmetric matrix formed from $\bZ^{\mathsf{mult}[2, 2]}$ by ``freezing'' one index pair $(i, j)$ and letting the other index pair vary, $\bF^{(i, j)} \colonequals (Z_{(i, j)(k,\ell)}^{\mathsf{mult}[2, 2]})_{k, \ell = 1}^N$, has row (or column) space contained in that of $\bM$.
Every $\bZ^{\mathsf{mult}[2, 2]}$ constructed as above must have this property, because it has a principal minor of the form
\begin{equation}
    \left[\begin{array}{cc} \bM & \bF^{(i, j)} \\ \bF^{(i, j)} & \bM \end{array}\right]
\end{equation}
whose positive semidefiniteness gives this condition on $\bF^{(i, j)}$.
Property 2 ensures that $Z_{(i,i)(k,\ell)}^{\mathsf{mult}[2,2]}$ does not depend on the index $i$, which is required by our definition of $\bZ^{\mathsf{mult}[2,2]}$ above, and reflects the application of the polynomial constraint $x_i^2 = 1$ to the monomial $x_i^2x_kx_{\ell}$.

What is the law of the resulting gaussian matrix $\bA$?
Conditioning on Property 1 yields the law of $\bP\bA^{(0)}\bP = \bV(\bV^\top \bA^{(0)}\bV)\bV^\top$.
By rotational invariance of the GOE, the inner matrix $\bV^\top \bA^{(0)} \bV \equalscolon \bA^{(1)} \in \RR^{r \times r}_{\sym}$ has the same law as the upper left $r \times r$ block of $\bA^{(0)}$, i.e., a smaller GOE matrix with the same variance scaling of $\sigma^2$.

Next, we condition on Property 2, or equivalently condition $\bA^{(1)}$ on having $\bv_i^\top \bA^{(1)} \bv_i = \la \bv_i\bv_i^\top, \bA^{(1)} \ra = 1$.
$\bA^{(1)}$ has the law of $\isovec^{-1}(\ba)$ for a gaussian vector $\ba \sim \sN(\bm 0, 2\sigma^2 \bm I_{r(r + 1) / 2})$.
Since $\isovec$ is an isometry, we may equivalently condition $\ba$ on $\la \ba, \isovec(\bv_i\bv_i^\top) \ra = 1$ for each $i \in [N]$.
By basic properties of gaussian conditioning, the resulting law is
\begin{equation}
\label{eq:gauss-heuristic-cond-law}
\sN\left(\sum_{i = 1}^N ((\bP^{\circ 2})^{-1} \one)_i\, \isovec(\bv_i\bv_i^\top), 2\sigma^2 (\bm I - \widetilde{\bm P})\right),
\end{equation}
where $\bP^{\circ 2}$ is the Gram matrix of the $\isovec(\bv_i\bv_i^\top)$ or equivalently the entrywise square of $\bP$, and $\widetilde{\bm P}$ is the orthogonal projector to the span of the $\isovec(\bv_i\bv_i^\top)$.
Let $\bA^{(2)}$ be a matrix with the law of $\isovec^{-1}$ applied to the law in \eqref{eq:gauss-heuristic-cond-law}.

Having finished the conditioning calculations, we may now obtain the statistics of $\bA$.
Recall that $A_{ij} = \bv_i^\top \bA^{(2)} \bv_j = \la \frac{1}{2}(\bv_i\bv_j^\top + \bv_j\bv_i^\top), \bA^{(2)} \ra$.
Applying $\isovec$ to each matrix and using the expression derived above, we find the mean and covariance
\begin{align}
    \EE[A_{ij}] &= \sum_{m = 1}^N ((\bP^{\circ 2})^{-1} \one)_k P_{im}P_{jm}, \\
    \mathsf{Cov}[A_{ij}, A_{k\ell}] &= \frac{\sigma^2}{2}\isovec(\bv_i\bv_j^\top + \bv_j\bv_i^\top)^\top (\bm I - \widetilde{\bm P})\hspace{0.1cm}\isovec(\bv_k\bv_\ell^\top + \bv_\ell\bv_k^\top). \label{eq:gauss-heuristic-cov-1}
\end{align}
Next, we make two simplifying approximations.
For the means, we approximate
\begin{equation}
    \bP^{\circ 2} \approx \delta \bm I + \frac{\delta}{r}\one_N\one_N^\top,
\end{equation}
which gives
\begin{equation}
    \EE[A_{ij}] \approx \delta^{-1}P_{ij} = M_{ij}.
\end{equation}
For the covariances, since under our assumptions we have  $\|\isovec(\bv_i\bv_i^\top)\|_2 = \|\bv_i\bv_i^\top\|_F = \|\bv_i\|_2^2 = \delta$, we approximate
\begin{equation}
    \label{eq:gauss-heuristic-proj}
    \widetilde{\bP} \approx \delta^{-2}\sum_{i = 1}^N \isovec(\bv_i\bv_i^\top)\isovec(\bv_i\bv_i^\top)^\top,
\end{equation}
which gives
\begin{equation}
    \label{eq:gauss-heuristic-cov}
    \mathsf{Cov}[A_{ij}, A_{k\ell}] \approx \sigma^2\left(P_{ik}P_{j\ell} + P_{i\ell}P_{jk} - 2\delta^{-2}\sum_{m = 1}^N P_{im}P_{jm}P_{km}P_{\ell m}\right).
\end{equation}

Finally, to recover what this prediction implies for the entries of $\bZ^{\mathsf{mult}[2,2]}$, we compute
\begin{align}
  Z_{(i,j)(k,\ell)}^{\mathsf{mult}[2,2]} &= \EE[A_{ij}A_{k\ell}] \nonumber \\
  &= \EE[A_{ij}]\EE[A_{k\ell}] + \mathsf{Cov}[A_{ij}, A_{k\ell}] \nonumber \\
  &= M_{ij}M_{k\ell} + \sigma^2\left(P_{ik}P_{j\ell} + P_{i\ell}P_{jk} - 2\delta^{-2}\sum_{m = 1}^N P_{im}P_{jm}P_{km}P_{\ell m}\right).
\end{align}
We then choose $\sigma^2$ such that $Z_{(i,i)(i,i)}^{\mathsf{mult}[2,2]} = 1$, which requires $\sigma^2 = \delta^{-2}$, and, restricting to $i \neq j$ and $k \neq \ell$, we recover the same formula as \eqref{eq:pm-pred-1}:
\begin{equation}
    \label{eq:pm-pred-2}
    \text{``}Z_{\{i,j\}\{k,\ell\}}^{[2,2]} = M_{ij}M_{k\ell} + M_{ik}M_{j\ell} + M_{i\ell}M_{jk} - 2 \sum_{m = 1}^N M_{im}M_{jm}M_{km}M_{\ell m}.\text{''}
\end{equation}

\begin{remark}
    It is worth noting the intriguing geometric interpretation of the random matrix $\bA$ we have constructed: we have $\EE \bA = \bM$, $\diag(\bm A) = \one$ deterministically, and $\bA$ fluctuates in the linear subspace $\mathsf{pert}_{\fE_2^N}(\bM)$ (as may be verified from the covariance formula \eqref{eq:gauss-heuristic-cov-1} and is intuitive by analogy with the ETF case of Section~\ref{sec:heuristic-etf}).
    Thus, $\bA$ behaves, roughly speaking, like a random element of $\fE_2^N$ (except that there is no enforcement of positive semidefiniteness), which lies on the same face of $\fE_2^N$ as $\bM$ and fluctuates gaussianly about $\bM$ along this face.
\end{remark}

\subsection{Precise construction details}

\label{sec:construction-details}

Having intuitively motivated the \emph{a priori} unusual degree 4 pseudomoment formula given (identically) in \eqref{eq:pm-pred-1} and \eqref{eq:pm-pred-2} in the previous two sections, we now give the precise details of how this may be adjusted to produce an actually valid degree 4 pseudomoment extension of $(1 - \alpha)\bM + \alpha \bm I_N$, the ``nudged'' Montanari-Sen witness.
It is instructive to view our construction as first attempting to build an extension of $\bM$ itself, then introducing the adjustment towards $\bm I_N$ as a necessity to ensure positive semidefiniteness.

\paragraph{Step 1: Heuristic pseudomoments}

We first build $\bX \in \RR^{\binom{[N]}{\leq 2} \times \binom{[N]}{\leq 2}}$ that is a reasonable prediction of a reduced degree 4 pseudomoment extension of $\bM$.
Viewing $\bX$ as a $3 \times 3$ block matrix as in \eqref{eq:deg-4-pm-blocks}, all blocks but the lower right are prescribed by the properties of being reduced and extending $\bM$, so we have
\begin{equation}
    \bX = \left[\arraycolsep 5pt \begin{array}{ccc} 1 & \bm 0 & \offdiag(\bM)^{\top} \\ \bm 0 & \bM & \bm 0 \\ \offdiag(\bM) & \bm 0 & \bX^{[2, 2]} \end{array}\right].
\end{equation}
We complete the definition by defining $\bX^{[2,2]}$ using the heuristics described earlier.
Namely, we take
\begin{equation}
    X^{[2, 2]}_{\{i,j\}\{k,\ell\}} \colonequals M_{ij}M_{k\ell} + M_{ik}M_{j\ell} + M_{i\ell}M_{jk} - 2\sum_{m = 1}^N M_{im}M_{jm}M_{km}M_{\ell m}.
\end{equation}
Reviewing the constraints required of $\bX$ per Proposition~\ref{prop:deg-4-pm-def}, we see that $\bX$ satisfies the permutation symmetry constraints (Condition 5 in the Proposition) and the normalization and reduction constraints (Conditions 2 and~3, respectively) exactly, and satisfies the other linear constraints which require $X^{[2, 2]}_{\{i,j\}\{i,k\}} = M_{jk}$ (Condition 4) approximately.
Finally, the discussion of Sections~\ref{sec:heuristic-etf} and \ref{sec:heuristic-gauss} suggests that $\bX$ should also be positive semidefinite (Condition 1).

\paragraph{Step 2: Correction to satisfy linear constraints}

We next correct $\bX$ to satisfy exactly all linear constraints required for a degree 4 pseudomoment matrix, by adjusting $\bX^{[2, 2]}$ to satisfy Condition 4 of Proposition~\ref{prop:deg-4-pm-def}.
Define an additive correction $\bm \Delta \in \RR^{\binom{[N]}{2} \times \binom{[N]}{2}}$ by
\begin{align}
  \Delta_{\{i,j\}\{k,\ell\}} &= 0 \text{ if } |\{i, j, k, \ell\}| = 4, \\
  \Delta_{\{i,k\}\{i,\ell\}} &= \sum_{\substack{m = 1 \\ m \neq i}}^N M_{im}^2M_{km}M_{\ell m}.
\end{align}
(Note that the second part of the definition is consistent when $k = \ell$ regardless of whether we view $i$ or $k$ as the repeated index.)

Then, we set
\begin{equation}
    \bY \colonequals \bX + \left[\arraycolsep 5pt \begin{array}{ccc} 0 & \bm 0 & \bm 0 \\ \bm 0 & \bm 0 & \bm 0 \\ \bm 0 & \bm 0 & 2\bm \Delta \end{array}\right],
\end{equation}
and $\bY$ satisfies Conditions 2 through 5 of Proposition~\ref{prop:deg-4-pm-def} exactly.
However, as we will see, $\bX^{[2, 2]}$ is a low-rank matrix, while $\bm\Delta$ acts non-trivially on components of the null space of $\bX^{[2, 2]}$.
Therefore, even if $\bX \succeq \bm 0$ (which, as we will see, is nearly true), we will still have $\bY \not\succeq \bm 0$ due to the fluctuations in $\bm\Delta$.

\paragraph{Step 3: Correction to satisfy positive semidefiniteness}

Finally, we introduce a second correction to counteract the fluctuations in the spectra of $\bX^{[2, 2]}$ and $\bm\Delta$.
We note that the identity matrix $\bm I_{\binom{[N]}{\leq 2}}$ is in fact a valid degree 4 pseudomoment matrix, which extends $\bm I_N \in \fC^N$.
Indeed, $\bm I_N$ is a natural choice of a point of $\fC^N$ towards which to ``push'' $\bM$ in order to regularize our construction, because $\bm I_N$ is the barycenter of the vertices of $\fC^N$, i.e., $\bm I_N = \frac{1}{2^N} \sum_{\bx \in \{ \pm 1\}^N} \bx\bx^{\top}$.

Following this intuition, given a choice of the parameter $\alpha \in (0, 1)$, we set
\begin{align}
    \bm Z
    &\colonequals (1 - \alpha)\bY + \alpha \bm I_{\binom{[N]}{\leq 2}} \nonumber \\
    &= \left[\arraycolsep 2pt \begin{array}{ccc} 1 & \bm 0 & (1 - \alpha)\offdiag(\bM)^{\top} \\ \bm 0 & (1 - \alpha)\bM + \alpha \bm I_N & \bm 0 \\ (1 - \alpha)\offdiag(\bM) & \bm 0 & (1 - \alpha)\bY^{[2, 2]} + \alpha \bm I_{\binom{[N]}{2}} \end{array}\right].
\end{align}
Clearly, $\bZ$ extends $(1 - \alpha)\bM + \alpha \bm I_N \in \fE_2^N$ and satisfies all linear constraints on a degree 4 pseudomoment matrix (since both $\bY$ and $\bm I_{\binom{[N]}{\leq 2}}$ do so).
Thus to show $(1 - \alpha)\bM + \alpha \bm I_N \in \fE_4^N$, it suffices to show $\bZ \succeq \bm 0$, in which case $\bZ$ will be a degree 4 pseudomoment extension.
Theorem~\ref{thm:sos4-witness} will then be proved if we show that $\bZ \succeq \bm 0$ with high probability.

\section{Conjectural higher-degree extension}

\label{sec:conj-extension}

Before proceeding to the proofs, we mention that there is a natural extension of the heuristic for pseudomoment construction of Section~\ref{sec:heuristic-gauss} that appears promising, though difficult to analyze, for higher-degree SOS relaxations.

The idea is to view higher-order pseudomoments as the second moments of symmetric gaussian \emph{tensors}, which, as done in Section~\ref{sec:heuristic-gauss} for matrices, are formed by conditioning a certain canonical symmetric tensor distribution on desirable properties.
Suppose we want to predict the degree $d = 2k$ pseudomoment extension $\bZ^{\mathsf{mult}[k, k]} \in \RR^{[N]^k \times [N]^k}$ of the Montanari-Sen witness $\bM \in \fE_2^N$ (which, as before, we assume to be an exact unit norm tight frame, i.e., a constant multiple of $\bP$), where, extending the case $k = 2$ from Section~\ref{sec:heuristic-gauss}, the entry indexed by $\bm i, \bm j \in [N]^k$ is the pseudoexpectation of $\prod_{\ell = 1}^kx_{i_{\ell}} \prod_{\ell = 1}^k x_{j_{\ell}}$.
We do this by building a symmetric tensor $\bA^{(k)} \in \mathsf{Sym}^k(\RR^N)$ with jointly gaussian entries, and setting, for $\bm i, \bm j \in [N]^k$,
\begin{equation}
    Z_{\bm i \bm j}^{\mathsf{mult}[k, k]} = \EE[A_{\bm i}^{(k)}A_{\bm j}^{(k)}].
\end{equation}

To describe the law of $\bA^{(k)}$, we first define the following tensorial generalization of the GOE (see, e.g.,~\cite{MR-tensor} for properties of this distribution analogous to those of the GOE).
\begin{definition}
    Let $\bG \in (\RR^N)^{\otimes k}$ have i.i.d.\ entries distributed as $\sN(0, 1)$.
    Then, write $\sG^{N, k}(\sigma^2)$ for the law of $\bA \in \mathsf{Sym}^k(\RR^N) \subset (\RR^N)^{\otimes k}$ defined by
    \begin{equation} A_{\bi} = \frac{\sigma}{k!}\sum_{\pi \in S_k} G_{i_{\pi(1)}i_{\pi(2)} \cdots i_{\pi(k)}}. \end{equation}
\end{definition}
\noindent
Now, we define $\bA^{(k)}$ inductively over $k$ as a family of \emph{coupled} gaussian tensors, and ensure that the pseudomoment matrices thus formed are consistent with one another.
Namely, we proceed as follows.
Let $\bm i \circ \bm j$ denote the concatenation of finite strings in the alphabet $[N]$.
\begin{enumerate}
    \item Let $A^{(0)}_{\emptyset} = 1$.
    \item For $k \geq 1$, let $\bA^{(k)}$ have the law $\sG^{N, k}(\sigma_k^2)$, conditioned on the following two properties:
        \begin{itemize}
        \item (Subspace Property) For $\bm i \in [N]^{k - 1}$, let $\bA^{(k)}[\bm i] \colonequals (A_{\bm i \circ (j)}^{(k)})_{j = 1}^N \in \RR^N$.
        Then, for all $\bm i \in [N]^{k - 1}$, $(\bm I - \bm P)\bA^{(k)}[\bm i] = \bm 0$.
        \item (Consistency Property) For $\bm i \in [N]^{k - 2}$ and $j \in [N]$, $A^{(k)}_{\bm i \circ (jj)} = A^{(k - 2)}_{\bm i}$.
        \end{itemize}
    \end{enumerate}
\noindent
The constants $\sigma_k^2$ remain as free parameters to be tuned to ensure normalization, as in the case $k = 2$ from Section~\ref{sec:heuristic-gauss}.

Based on this reasonable generalization, we offer two conjectures.
First, we believe that whatever adjustments are necessary to this construction are already captured in the simple adjustment of the Montanari-Sen witness towards the identity matrix given in Theorem~\ref{thm:sos4-witness}.
\begin{conjecture}
    \label{conj:all-degrees-elliptope}
    For any $\alpha, \delta \in (0, 1)$ and $d \in 2\NN$,
    \begin{equation} \lim_{N \to \infty} \PP\left[ \bM^{(\alpha, \delta)}(\bW) \in \fE_d^N \right] = 1. \end{equation}
\end{conjecture}

\noindent
More specifically but less formally, we believe the construction presented above is approximately the correct degree $d$ pseudomoment extension.

\begin{conjecture}[Informal]
\label{conj:higher-degree}
    For $\delta \in (0, 1)$, let $\bM^{(\delta)}$ be the Montanari-Sen witness.
    Then, for fixed $d \in 2\NN$ and constants $\sigma_k^2$ depending only on $\delta$, with high probability as $N \to \infty$, the entries of $\bZ^{\mathsf{mult}[k, k]}$ as defined above give a ``nearly'' valid degree $d = 2k$ pseudomoment extension of $\bM^{(\delta)}$.
\end{conjecture}

We have tested Conjecture~\ref{conj:higher-degree} numerically on Laurent's construction \cite{Laurent-SOS} of higher-degree pseudomoment matrices extending a deterministic $\bM \in \fE_2^N$, which indeed forms the Gram matrix of an ETF.
Laurent's construction shows that certain parity inequalities holding over $\fC^N$ are not certified by SOS until $d \sim N$.
We find that the results, for suitable tuning of $\sigma_k^2$, agree with Laurent's construction (with no further adjustment needed).
Thus one pleasant consequence of verifying Conjecture~\ref{conj:higher-degree} may be a novel proof of Laurent's theorem, whose original proof involves first predicting the entries of the pseudomoment matrix and then appealing to a technical analysis of hypergeometric functions to verify positive semidefiniteness.

We remark that we have also verified algebraically in our earlier paper~\cite{BK-ETF} that, for degree 4, the construction in Theorem~\ref{thm:etf-extension} (before simplifying in the $N \to \infty$ asymptotic regime) exactly recovers Laurent's construction.
This also would not be the first simplified proof of Laurent's theorem; for instance, the work~\cite{KLM-2016-SymmetricFormulations} gives a general treatment of sum-of-squares relaxations of combinatorial optimization problems with highly symmetric formulations, and also produces a Cholesky-type decomposition of the relevant pseudomoment matrix.
However, the idea we outline above both would unify such results with those for less symmetric problems, and would give a natural interpretation of the Cholesky factors in such a decomposition as the random variables that are the entries of the tensors $\bA^{(k)}$.

Finally, let us remark on what seems to be the major difficulty in analyzing this construction.
By analogy with the analysis in Section~\ref{sec:heuristic-gauss}, we are eventually led, in conditioning on the Consistency Property, to attempt to approximate the orthogonal projector to the ``repeated indices subspace''
\begin{align}
  V^{(k)}
  &= \mathsf{span}\left(\left\{\bv_i \odot \bv_i \odot \bv_{j_1} \odot \cdots \odot \bv_{j_{k - 2}} : i \in [N], \bm j \in [N]^{k - 2} \right\}\right) \nonumber \\
  &= \mathsf{span}\left(\left\{\bv_i \odot \bv_i \odot \bm A: i \in [N], \bm A \in \mathsf{Sym}^{k - 2}(\RR^r) \right\}\right) \nonumber \\
  &\subset \mathsf{Sym}^k(\RR^r).
\end{align}
(Here $\odot$ denotes the symmetric product of tensors; see, e.g., \cite{SKM-tensor} for definitions. We mean ``orthogonal projection'' with respect to the Frobenius or entrywise inner product of general non-symmetric tensors, into which symmetric tensors are embedded by repeating entries.)
When $k = 2$, the spanning set consists of the $N$ linearly independent and roughly orthogonal tensors $\bv_i^{\odot 2}$, which allows the orthogonal projection to be estimated by the sum of rank one projections as in \eqref{eq:gauss-heuristic-proj}.
However, when $k \geq 3$, there does not appear to be a clear way to choose a convenient approximately-orthogonal basis to carry out the calculation.
The collection of tensors $\bv_i^{\odot 2} \odot \bv_{j_1} \odot \cdots \odot \bv_{j_{k - 2}}$ is highly overcomplete, since the $\bv_i$ themselves are an overcomplete set in $\RR^r$ (and the symmetric tensor product is distributive, so any dependence among the $\bv_i$ is inherited by the $\bA \odot \bv_i$ for any symmetric tensor $\bA$).
Moreover, even the subspaces $V^{(k)}_i \colonequals \bv_i^{\odot 2} \odot \mathsf{Sym}^{k - 2}(\RR^r)$ intersect non-trivially; for instance, $\bv_i^{\odot 2} \odot \bv_j^{\odot 2} \in V^{(4)}_i \cap V^{(4)}_j$.
Thus it appears that a deeper understanding of the structure of these subspaces of symmetric tensors is required to form the correct higher-degree analogue of \eqref{eq:gauss-heuristic-proj}.

\section{Proof of positive semidefiniteness: first steps}

\label{sec:proof-outline}

Recall that, in Section~\ref{sec:construction-details}, we built from the Montanari-Sen witness $\bM$ and an additional constant $\alpha \in (0, 1)$ the matrix
\begin{align}
    \bZ
    &= \left[\arraycolsep 4pt \begin{array}{ccc} \bZ^{[0, 0]} & \bZ^{[0, 1]} & \bZ^{[0, 2]} \\ \bZ^{[1, 0]} & \bZ^{[1, 1]} & \bZ^{[1, 2]} \\ \bZ^{[2, 0]} & \bZ^{[2, 1]} & \bZ^{[2, 2]} \end{array}\right] \nonumber \\
    &= \left[\arraycolsep 2pt \begin{array}{ccc} 1 & \bm 0 & (1 - \alpha)\offdiag(\bM)^{\top} \\ \bm 0 & (1 - \alpha)\bM + \alpha \bm I_N & \bm 0 \\ (1 - \alpha)\offdiag(\bM) & \bm 0 & (1 - \alpha)\bY^{[2, 2]} + \alpha \bm I_{\binom{[N]}{2}} \end{array}\right]
\end{align}
and found that to prove Theorem~\ref{thm:sos4-witness} it suffices to show that $\bZ \succeq \bm 0$ with high probability.
We now give some technical preliminaries for the proof of this.

First, note that, after permuting rows and columns, $\bZ$ is the direct sum of $(1 - \alpha)\bM + \alpha \bm I_N$, which is positive semidefinite by assumption, with the principal minor of $\bZ$ indexed by $\{\emptyset\} \cup \binom{[N]}{2}$.
Thus to show that $\bZ \succeq \bm 0$ it suffices to show that the latter minor is positive semidefinite.

Second, we may reduce the dimensionality of this remaining task by taking the Schur complement criterion for positive semidefiniteness with respect to the upper left entry, indexed by $(\emptyset, \emptyset)$, whose value is 1.
The condition of positive semidefiniteness of the Schur complement is then
\begin{align}
  &\bZ / \bZ^{[0, 0]} \nonumber \\
  &\hspace{0.5cm} = \bm Z^{[2, 2]} - \bm Z^{[2, 0]}\bm Z^{[0, 2]} \nonumber \\
  &\hspace{0.5cm} = \alpha \bm I_{N(N - 1)/2} + (1 - \alpha)\bY^{[2, 2]} - (1 - \alpha)^2\offdiag(\bM)\offdiag(\bM)^\top \nonumber \\
  &\hspace{0.5cm} = \alpha \bm I_{N(N - 1)/2} + (1 - \alpha)(\bX^{[2, 2]} + 2\bm\Delta) - (1 - \alpha)^2\offdiag(\bM)\offdiag(\bM)^\top \nonumber \\
    &\hspace{0.5cm} \stackrel{?}{\succeq} \bm 0.
\end{align}
We reorganize this expression as
\begin{align}
  \bZ / \bZ^{[0, 0]} &= \bZ^{(1)} + \bZ^{(2)}, \text{ where } \\
  \bZ^{(1)} &= \frac{1}{2}\alpha \bm I_{N(N - 1) / 2} \nonumber\\
  &\hspace{1cm} + (1 - \alpha)\big(\underbrace{\bX^{[2, 2]} - (1 - \alpha)\offdiag(\bM)\offdiag(\bM)^\top}_{\bZ^{(1a)}}\big), \\
  \bZ^{(2)} &= \frac{1}{2}\alpha \bm I_{N(N - 1) / 2} + 2(1 - \alpha)\bm\Delta.
\end{align}
To show $\bZ \succeq \bm 0$, it then suffices to show that both $\bZ^{(1)} \succeq \bm 0$ and $\bZ^{(2)} \succeq \bm 0$.
We refer to these as the ``main term'' and the ``correction term,'' respectively.

In this decomposition, we split between $\bZ^{(1)}$ and $\bZ^{(2)}$ the extra term $\alpha \bm I_{N(N - 1) / 2}$ that we introduced when nudging our pseudomoment matrix towards the identity.
This term will act as a ``barrier'' against small fluctuations that might spoil positive semidefiniteness.
We will show that without this adjustment $\bZ^{(1)}$ and $\bZ^{(2)}$ are nearly positive semidefinite already, having the magnitude of their smallest (most negative) eigenvalue tending to zero as $N \to \infty$ for any fixed $\delta$.
Thus any choice of $\alpha > 0$ will suffice to ensure that $\bZ \succeq \bm 0$ with high probability.

More specifically, we will show the following results.
\begin{lemma}[Control of main term]
    \label{lem:main-term}
    For all $\delta \in (0, 1)$,
    \begin{equation} \lim_{N \to \infty}\PP\left[|\min\{0, \lambda_{\min}(\bZ^{(1a)})\}| \leq O_{\delta}\left(\frac{\log N}{N^{1/4}}\right)\right] = 1. \end{equation}
\end{lemma}

\begin{lemma}[Control of correction term]
    \label{lem:corr-term}
    For all $\delta \in (0, 1)$,
    \begin{equation} \lim_{N \to \infty}\PP\left[\|\bm\Delta\|_{\op} \leq O_{\delta}\left(\frac{\log^2 N}{N^{1/4}}\right)\right] = 1. \end{equation}
\end{lemma}

\noindent
From Lemmata~\ref{lem:main-term} and \ref{lem:corr-term}, it follows that with high probability $\bZ \succeq \bm 0$ for any $\alpha \in (0, 1)$ fixed as $N \to \infty$ (or for $\alpha = \alpha(N)$ decreasing sufficiently slowly with $N$, though for the sake of simplicity we will not pursue this minor strengthening of the results).
Theorem~\ref{thm:sos4-witness} then follows.
It remains only to prove the Lemmata; in Section~\ref{sec:main-term} we will prove Lemma~\ref{lem:main-term}, and in Section~\ref{sec:corr-term} we will prove Lemma~\ref{lem:corr-term}.

\section{Proof of positive semidefiniteness: main term}

\label{sec:main-term}

In this section we prove Lemma~\ref{lem:main-term}.
We have
\begin{align}
  \hspace{0.5cm}&\hspace{-0.5cm}Z^{(1a)}_{\{i,j\}\{k,\ell\}} \nonumber \\
  &= X^{[2, 2]}_{\{i,j\}\{k,\ell\}} - (1 - \alpha)M_{ij}M_{k\ell} \nonumber \\
  &= \alpha M_{ij}M_{k\ell} + M_{ik}M_{j\ell} + M_{i\ell}M_{jk} - 2\sum_{m = 1}^N M_{im}M_{jm}M_{km}M_{\ell m}.
\end{align}
Consider the quadratic form $\ba^\top \bZ^{(1a)} \ba$, where we think of $a_{\{i,j\}} = 2A_{ij}$ for some symmetric matrix $\bA$ with $\diag(\bA) = \bm 0$ (i.e., $\ba = 2 \cdot  \offdiag(\bA)$).
Writing $\bbm_1, \dots, \bbm_N \in \RR^N$ for the columns of $\bM$,
\begin{align}
  \ba^\top \bZ^{(1a)} \ba
  &= \alpha \Tr(\bA \bM)^2 + 2\Tr(\bA\bM\bA\bM) - 2\sum_{i = 1}^N (\bbm_i^\top \bA \bbm_i)^2 \nonumber \\
  &= \alpha \Tr(\what{\bV}\bA \what{\bV}^\top)^2 + 2\|\what{\bV}\bA\what{\bV}^\top\|_F^2 - 2\sum_{i = 1}^N (\what{\bv}_i^\top \what{\bV}\bA\what{\bV}^\top \what{\bv}_i)^2. \label{eq:Z1a-reduction}
\end{align}
Writing the above as a quadratic form in $\isovec(\what{\bV}\bA\what{\bV}^\top)$, we obtain the following.
\begin{proposition}
    \label{prop:Z1a-tZ1a}
    Define
    \begin{equation}
        \widetilde{\bZ}^{(1a)} \colonequals \alpha \one_{\diag}\one_{\diag}^\top + 2 \bm I_{r(r + 1) / 2} - 2 \sum_{i = 1}^N \isovec(\what{\bv}_i\what{\bv}_i)\isovec(\what{\bv}_i\what{\bv}_i)^\top.
       \label{eq:tZ1a-def}
   \end{equation}
   (Recall that $\one_{\diag} \colonequals \isovec(\bm I_r)$.)
    Then,
    \begin{equation}
        \label{eq:Z1a-tilde-bound}
        \lambda_{\min}(\bZ^{(1a)}) \geq \min\left\{0, \frac{\lambda_{\min}(\widetilde{\bZ}^{(1a)})}{2\min_{i \in [N]}D_{ii}^2}\right\}.
    \end{equation}
\end{proposition}
\begin{proof}
    Since the right-hand side of \eqref{eq:Z1a-tilde-bound} is at most zero, it suffices to consider the case that $\lambda_{\min}(\bZ^{(1a)}) < 0$.
    By \eqref{eq:Z1a-reduction}, we have
    \begin{equation}
        \isovec(\bA)^\top\bZ^{(1a)}\isovec(\bA) = \isovec(\what{\bV}\bA\what{\bV}^\top)^\top \widetilde{\bZ}^{(1a)}\isovec(\what{\bV}\bA\what{\bV}^\top).
    \end{equation}
    Note first that
    \begin{align}
      \|\what{\bV}\bA\what{\bV}^\top\|_F^2
      &= \Tr(\bA\bM\bA\bM) \nonumber \\
      &= \Tr(\bA\bD^{-1/2}\bP\bD^{-1/2}\bA\bD^{-1/2}\bP\bD^{-1/2}) \nonumber \\
      &= \la \bP, \bD^{-1/2}\bA\bD^{-1/2}\bP\bD^{-1/2}\bA\bD^{-1/2} \ra \nonumber \\
      &\leq \Tr(\bD^{-1/2}\bA\bD^{-1/2}\bP\bD^{-1/2}\bA\bD^{-1/2}) \nonumber \\
      &\leq \Tr((\bD^{-1/2}\bA\bD^{-1/2})(\bD^{-1/2}\bA\bD^{-1/2})) \nonumber \\
      &= \|\bD^{-1/2}\bA\bD^{-1/2}\|_F^2 \nonumber \\
      &\leq \frac{\|\bA\|_F^2}{\min_{i \in [N]} D_{ii}^2}.
        \label{eq:norm-relation}
    \end{align}
    Then, recalling that $\|\ba\|_2^2 = 2\|\bA\|_F^2$ and $\|\isovec(\what{\bV}\bA\what{\bV}^\top)\|_2^2 = \|\what{\bV}\bA\what{\bV}^\top\|_F^2$, by the variational description of the minimum eigenvalue we have
    \begin{align}
      \hspace{0.5cm}&\hspace{-0.5cm}\lambda_{\min}(\bZ^{(1a)}) \nonumber \\
                    &= \min_{\ba \in \RR^{N(N - 1) / 2} \setminus\{ \bm 0\}} \frac{\ba^\top \bZ^{(1a)}\ba}{\|\ba\|_2^2} \nonumber
      \intertext{and noting that, since we assume $\lambda_{\min}(\bZ^{(1a)}) < 0$, the $\ba$ achieving the minimum has $\ba^{\top} \bZ^{(1a)}\ba < 0$, so we may continue}
      &\geq  \frac{1}{2\min_{i \in [N]} D_{ii}^2}\min_{\ba \in \RR^{N(N - 1) / 2} \setminus\{ \bm 0\}}\frac{\ba^\top \bZ^{(1a)}\ba}{\|\what{\bV}\bA\what{\bV}^\top\|_F^2} \tag{by \eqref{eq:norm-relation}} \\
      &=  \frac{1}{2\min_{i \in [N]} D_{ii}^2}\min_{\ba \in \RR^{N(N - 1) / 2} \setminus\{ \bm 0\}} \frac{\isovec(\what{\bV}\bA\what{\bV}^\top)^\top\widetilde{\bZ}^{(1a)}\isovec(\what{\bV}\bA\what{\bV}^\top)}{\|\isovec(\what{\bV}\bA\what{\bV}^\top)\|_2^2} \nonumber \\
      &\geq \frac{\lambda_{\min}(\widetilde{\bZ}^{(1a)})}{2\min_{i \in [N]} D_{ii}^2},
    \end{align}
    completing the proof.
\end{proof}

We will thus focus our attention on $\widetilde{\bZ}^{(1a)}$.
Analyzing the Wishart-type matrix formed by the third term of \eqref{eq:tZ1a-def}, $\sum_{i = 1}^N \isovec(\what{\bv}_i\what{\bv}_i)\isovec(\what{\bv}_i\what{\bv}_i)^\top$, will be our main difficulty.
Since $\EE\what{\bv}_i\what{\bv}_i^\top = \frac{1}{r}\bm I_r$, we center the vectors involved, and decompose this term as
\begin{align}
  \widetilde{\bZ}^{(1a)} &= \left(\alpha - \frac{2N}{r^2}\right) \one_{\diag}\one_{\diag}^\top + 2 \bm I_{r(r + 1) / 2} \nonumber \\
                         &\hspace{-0.5cm}+ \underbrace{\frac{2}{r}\left(\isovec\left(\sum_{i = 1}^N \what{\bv}_i\what{\bv}_i^\top - \frac{N}{r}\bm I_r\right)\one_{\diag}^\top + \one_{\diag}\isovec\left(\sum_{i = 1}^N \what{\bv}_i\what{\bv}_i^\top - \frac{N}{r}\bm I_r\right)^\top\right)}_{\bT^{(1)}} \nonumber \\
  &\hspace{-0.5cm}- 2\underbrace{\sum_{i = 1}^N \isovec\left(\what{\bv}_i\what{\bv}_i - \frac{1}{r}\bm I_r\right)\isovec\left(\what{\bv}_i\what{\bv}_i - \frac{1}{r}\bm I_r\right)^\top}_{\bT^{(2)}}.
\end{align}
The remaining analysis involves a delicate balance between requirements in controlling $\bT^{(1)}$ and $\bT^{(2)}$.
In order to bound the cross-term $\bT^{(1)}$, we will rely on the strong concentration of the eigenvalues of $\what{\bV}\what{\bV}^\top$ that is created by the dependencies among the $\what{\bv}_i$ (this is the ``near-UNTF Gram matrix'' behavior of $\bM$).
In particular, this concentration is much stronger than if $\what{\bv}_i$ were replaced with any reasonable distribution of i.i.d.\ unit vectors, and this portion of our argument would fail for i.i.d.\ vectors (see Remark~\ref{rem:iid-vs-orth}).

On the other hand, in order to bound the term $\bT^{(2)}$, we will need to take advantage of the weak dependence of the $\what{\bv}_i$, and formalize the intuition that because $N \ll r(r + 1) / 2$ and $\bT^{(2)}$ is a sum of weakly dependent rank-one orthogonal projectors, $\bT^{(2)}$ should itself behave approximately as an orthogonal projector to a subspace of dimension $N$ (though we will discuss one important caveat to this intuition in Remark~\ref{rem:iid-vs-orth}).
Technically, we will appeal to Lipschitz concentration inequalities for the Haar measure on Stiefel manifolds, which capture the heuristic weak dependence of entries of blocks of random orthogonal matrices under the Haar measure.

\subsection{Bounding the cross-term \texorpdfstring{$\bmrob T^{(1)}$}{T1}}

\begin{lemma}
    \label{lem:cross-term}
    For all $K > 0$,
    \begin{equation}
        \PP\left[|\bT^{(1)}| \preceq O_{\delta, K}\left(\frac{\log N}{N}\right) \bm I_{r(r + 1) / 2} + \frac{2}{r}\one_{\diag}\one_{\diag}^\top\right] \geq 1 - O_{\delta, K}(N^{-K}).
    \end{equation}
\end{lemma}
\begin{proof}
    Applying the matrix arithmetic-geometric mean inequality,
    \begin{align}
      |\bT^{(1)}|
      &\preceq \frac{2}{r}\isovec\left(\sum_{i = 1}^N \what{\bv}_i\what{\bv}_i^\top - \frac{N}{r}\bm I_r\right)\isovec\left(\sum_{i = 1}^N \what{\bv}_i\what{\bv}_i^\top - \frac{N}{r}\bm I_r\right)^\top \nonumber \\ &\hspace{1cm} + \frac{2}{r}\one_{\diag}\one_{\diag}^{\top} \nonumber \\
      &\preceq \frac{2}{r}\left\|\sum_{i = 1}^N \what{\bv}_i\what{\bv}_i^\top - \frac{N}{r}\bm I_r\right\|_F^2 \bm I_{r(r + 1) / 2} + \frac{2}{r}\one_{\diag}\one_{\diag}^{\top}.
    \end{align}
    Rewriting the norm appearing in the first term,
    \begin{align}
      \left\|\sum_{i = 1}^N \what{\bv}_i\what{\bv}_i^\top - \frac{N}{r}\bm I_r\right\|_F^2
      &= \left\|\what{\bV}\what{\bV}^\top - \frac{N}{r}\bm I_r\right\|_F^2 \nonumber \\
      &= \left\|\bV\bD^{-1}\bV^\top - \delta^{-1}\bm I_r\right\|_F^2 \nonumber \\
      &= \left\|\bD^{-1} - \delta^{-1}\bm I_N\right\|_F^2, \end{align}
    since $\bV\bV^\top = \bm I_r$.
    By Proposition~\ref{prop:proj-entries},
    \begin{align}
        &\PP\left[\left(\delta^{-1} - O_{\delta, K}\left(\sqrt{\frac{\log N}{N}}\right)\right)\bm I_{N} \preceq \bD^{-1} \preceq \left(\delta^{-1} + O_{\delta, K}\left(\sqrt{\frac{\log N}{N}}\right)\right)\bm I_{N}\right] \nonumber \\ &\hspace{1cm}\geq 1 - O_{\delta, K}(N^{-K}).
    \end{align}
    Thus with at least the same probability we have
    \begin{equation}
    \left\|\bD^{-1} - \delta^{-1}\bm I_N\right\|_F^2 \leq O_{\delta, K}(\log N), \end{equation}
and the result follows.
\end{proof}

\begin{remark}
    \label{rem:iid-vs-orth}
    Let us contrast the result of this section with the same analysis for i.i.d.\ vectors.
    The marginal law of each $\what{\bv}_i$ is uniform over $\SS^{r - 1}$, so consider taking $\widetilde{\bv}_i \sim \mathsf{Unif}(\SS^{r - 1})$ independent.
    Then, we compute
    \begin{align}
      \EE\left\|\sum_{i = 1}^N \widetilde{\bv}_i\widetilde{\bv}_i^\top - \frac{N}{r}\bm I_r\right\|_F^2
      &= N \EE \|\widetilde{\bv}_1\|_2^4 + N(N - 1)\EE\la \widetilde{\bv}_1, \widetilde{\bv}_2 \ra^2 - \frac{N^2}{r} + N \nonumber \\
      &= \Omega(N). \end{align}
    Thus the corresponding cross-term would have largest eigenvalue of order $\delta^{-1}$, which in particular would not decay with $N$.

    Consequently, our previous intuition that we should obtain an approximate projector of rank $N$ from $\bT^{(2)} = \sum_{i = 1}^N\isovec(\what{\bv}_i\what{\bv}_i^\top - \frac{1}{r}\bm I_r)\isovec(\what{\bv}_i\what{\bv}_i^\top - \frac{1}{r}\bm I_r)^\top$ cannot be correct, since the putative basis vectors $\isovec(\what{\bv}_i\what{\bv}_i^\top - \frac{1}{r}\bm I_r)$ almost sum to zero.
    In the following section, we will show that this is in fact the only linear near-dependence of these vectors, and $\bT^{(2)}$ is still an approximate orthogonal projector, only of rank $N - 1$.
\end{remark}

\subsection{Bounding the projection term \texorpdfstring{$\bmrob T^{(2)}$}{T2}: unnormalized case}

\label{sec:unnormalized}

Our strategy for bounding $\bT^{(2)}$ will proceed in two steps: first, we will bound the same matrix but constructed from the approximately normalized vectors $\delta^{-1/2}\bv_1, \dots, \delta^{-1/2}\bv_N$ in place of the strictly normalized vectors $\what{\bv}_1, \dots, \what{\bv}_N$, and then we will show that this replacement does not significantly affect the spectrum.
In this section we perform the first, more difficult of these tasks.
We will show the following result.

\begin{lemma}
    \label{lem:orth-sym-cov}
    Let $\bA^{\orth}$ have $\isovec(\delta^{-1}\bv_i\bv_i^\top - \frac{1}{r}\bm I_r)$ as its columns.
    Let $\bP_{\one_N^\top} \colonequals \bm I_N - \frac{1}{N}\one_N\one_N^\top$, the orthogonal projector to the subspace orthgonal to $\one_N$.
    Then,
    \begin{equation}
        \PP\left[ \|\bA^{\orth^\top}\bA^{\orth} - \bP_{\one_N^\top}\|_{\op} \leq O_{\delta}\left(\frac{\log N}{N^{1/4}}\right) \right] \geq 1 -  \exp\left(-\Omega_{\delta}(N^{1/2})\right).
    \end{equation}
\end{lemma}
\noindent
The argument will use the technique of union bounding over a net.
Our main technical tool will be the following Lipschitz concentration inequality for the Haar measure of the \emph{Stiefel manifolds}.

Recall that the Stiefel manifolds are defined as:
\begin{equation} \mathsf{Stief}(N, r) \colonequals \{ \bV \in \RR^{r \times N}: \bV \bV^\top = \bm I_r \}. \end{equation}
The Haar measure $\Haar(\mathsf{Stief}(N, r))$ may be viewed as the measure obtained by restricting $\Haar(\sO(N))$ to the upper $r \times N$ matrix block.

These measures enjoy the following concentration inequality when $r < N$, obtained by standard arguments from logarithmic Sobolev or isoperimetric inequalities for the special orthogonal group $\mathcal{SO}(N)$, of which $\mathsf{Stief}(N, r)$ is a quotient when $r < N$ (see, e.g., the discussion following Theorem~2.4 of \cite{Ledoux-conc}).
\begin{proposition}
    \label{prop:orth-lip-conc}
    Suppose $1 \leq r < N$, and $F: \mathsf{Stief}(N, r) \to \RR$ has Lipschitz constant at most $L$ when $\mathsf{Stief}(N, r)$ is endowed with the metric of the Frobenius matrix norm.
    Then, for an absolute constant $C > 0$,
    \begin{equation}
        \label{eq:op-norm-bd-proj}
        \PP_{\bV \sim \Haar(\mathsf{Stief}(N, r))}\left[\left|F(\bV) - \EE F(\bV)\right| \geq t\right] \leq 2\exp\left(-\frac{CNt^2}{L^2}\right).
    \end{equation}
\end{proposition}
\noindent
Note that since we have $r = \delta N$ with $\delta < 1$, we will always satisfy the hypothesis $r < N$; we will use this implicitly without further comment for the remainder of the proof.

\begin{proof}[Proof of Lemma~\ref{lem:orth-sym-cov}]
    Note that $\bA^{\orth}\one_N = \isovec(\frac{N}{r}\bV\bV^\top - \frac{N}{r}\bm I_r) = \bm 0$; thus as suggested already in Remark~\ref{rem:iid-vs-orth}, it is impossible for $\bA^{\orth}$ to act on $\RR^N$ as an approximate isometric embedding, as we might naively expect from its weakly dependent columns.
    Our argument is more natural to carry out if we remove this caveat; therefore, let us define $\bA^\orth_0$ to have columns $\isovec(\frac{N}{r}\bv_i\bv_i^\top - \frac{1 - \sqrt{\delta}}{r}\bm I_r)$.
    One may check that $\|\bA^\orth_0\one_N\|_2 = \|\one_N\|_2 = \sqrt{N}$, and that
    \begin{equation} \bA_0^{\orth^\top}\bA_0^{\orth} = \bA^{\orth^\top}\bA^\orth + \frac{1}{N}\one_N\one_N^\top. \end{equation}
    In particular, $\bA_0^{\orth^\top}\bA_0^\orth - \bm I_N = \bA^{\orth^\top}\bA^\orth - \bP_{\one_N^\perp}$, so it suffices to show the operator norm bound of \eqref{eq:op-norm-bd-proj} for $\bA_0^{\orth^\top}\bA_0^\orth - \bm I_N$.

    For $\bx \in \RR^N$, let us denote $\bD_{\bx} \colonequals \diag(\bx)$ for the course of this proof.
    Then,
    \begin{align}
      \bA^\orth_0 \bx
      &= \isovec\left(\delta^{-1}\sum_{i = 1}^N x_i \bv_i\bv_i^\top - \frac{1 - \sqrt{\delta}}{r}\la \one_N, \bx \ra\bm I_r\right) \nonumber \\
      &= \delta^{-1}\isovec\left(\bV\bD_{\bx}\bV^\top - \frac{1 - \sqrt{\delta}}{N}\la \one_N, \bx \ra \bm I_r\right).
    \end{align}
    For $\bx, \by \in \RR^N$, define
    \begin{align}
      F_{\bx, \by}(\bV) &\colonequals \la \bA_0^\orth\bx, \bA_0^\orth\by \ra.
    \end{align}
    Then, recalling that $\bP \colonequals \bV^\top \bV$ is the orthogonal projector to the row space of $\bV$,
    \begin{align}
      \hspace{0.5cm}&\hspace{-0.5cm}F_{\bx, \by}(\bV) \nonumber \\
      &= \delta^{-2}\left \la \bV\bD_{\bx}\bV^\top - \frac{1 - \sqrt{\delta}}{N}\la \one_N, \bx \ra \bm I_r, \bV\bD_{\by}\bV^\top - \frac{1 - \sqrt{\delta}}{N}\la \one_N, \by \ra \bm I_r \right\ra \nonumber \\
      &= \delta^{-2}\bigg[\Tr\left(\bD_{\bx}\bP\bD_{\by}\bP\right) \nonumber \\
      &\hspace{2cm} - \frac{1 - \sqrt{\delta}}{N}\big(\la \one_N, \bx \ra\Tr(\bP\bD_{\by}) + \la \one_N, \by \ra \Tr(\bP\bD_{\bx})\big) \nonumber \\
      &\hspace{2cm} + \frac{\delta(1 - \sqrt{\delta})^2}{N}\la \one_N, \bx \ra\la \one_N, \by \ra\bigg]. \label{eq:Fxy-expansion}
    \end{align}
    Let us denote balls in Euclidean space by
    \begin{equation}
        B(\bm x, r) \colonequals \left\{ \by \in \RR^N : \|\bx - \by\|_2 \leq r\right\}.
    \end{equation}
    Our first goal will be to obtain concentration bounds on $F_{\bx, \by}(\bV)$ when $\bV \sim \Haar(\mathsf{Stief}(N, r))$ for each fixed pair $(\bx, \by) \in B(\bm 0, 1)^2$, by applying the Lipschitz concentration inequality.

    \begin{claim}
        \label{claim:lip}
        Let $\bx, \by \in B(\bm 0, 1)$. Then,
        \begin{equation} \mathsf{Lip}(F_{\bx, \by}) \leq 4\delta^{-2}\left[\min\left\{\|\bx\|_\infty, \|\by\|_\infty\right\} + \frac{1}{\sqrt{N}}\right]. \end{equation}
    \end{claim}
    \begin{proof}
        For $\bV_1, \bV_2 \in \mathsf{Stief}(N, r)$, letting $\bP_i = \bV_i^\top \bV_i$, we have using \eqref{eq:Fxy-expansion} and the triangle inequality
        \begin{align}
          &\delta^2\left|F_{\bx, \by}(\bV_1) - F_{\bx, \by}(\bV_2)\right| \nonumber \\
          &\hspace{1cm} = \left|\Tr\left(\bD_{\bx}\bP_1\bD_{\by} \bP_1\right) - \Tr\left(\bD_{\bx}\bP_2\bD_{\by} \bP_2\right)\right\| \nonumber \\
          &\hspace{2cm} + \frac{(1 - \sqrt{\delta})|\la \one_N, \bx\ra|}{N}\left| \Tr(\bP_1\bD_y) - \Tr(\bP_2\bD_y)\right| \nonumber \\
          &\hspace{2cm} + \frac{(1 - \sqrt{\delta})|\la \one_N, \by\ra|}{N}\left| \Tr(\bP_1\bD_{\bx}) - \Tr(\bP_2\bD_{\bx})\right|, \nonumber \\
          \intertext{then using that $|\la \one_N, \bx \ra| \leq \|\bx\|_1 \leq \sqrt{N}$ and likewise for $\by$,}
          &\hspace{1cm}\leq \left|\Tr\left(\bD_{\bx}\bP_1\bD_{\by} \left(\bP_1 - \bP_2\right)\right)\right| + \left|\Tr\left(\bD_{\bx}\left(\bP_1 - \bP_2\right)\bD_{\by} \bP_2\right)\right| \nonumber \\
          &\hspace{2cm}+ \frac{1}{\sqrt{N}}\left(|\Tr((\bP_1 - \bP_2)\bD_{\bx})| + |\Tr((\bP_1 - \bP_2)\bD_{\by})|\right) \nonumber \\
          &\hspace{1cm}\leq (\|\bD_{\bx}\bP_1\|_F + \|\bD_{\bx}\bP_2\|_F)\|\bD_{\by}(\bP_1 - \bP_2)\|_F \nonumber
          \\
          &\hspace{2cm}+ \frac{2}{\sqrt{N}}\|\bP_1 - \bP_2\|_F.
        \end{align}
        Since $\bP_i$ is an orthogonal projector for $i \in \{1, 2\}$,
        \begin{equation}
          \left\|\bD_{\bx}\bP_i\right\|_F = \left \la \bD_{\bx}^2, \bP_i \right \ra^{1/2} \leq (\Tr[\bD_{\bx}^2])^{1/2} \leq 1.
        \end{equation}
        We bound the other term by
        \begin{equation}
          \left\|\bD_{\by}\left(\bP_1 - \bP_2\right)\right\|_F = \left \la \bD_{\by}^2, (\bP_1 - \bP_2)^2 \right \ra^{1/2} \leq \|\by\|_{\infty}\|\bP_1 - \bP_2\|_F.
      \end{equation}
      Combining these observations and a symmetric argument with $\bx$ and $\by$ in opposite roles gives
      \begin{equation} |F_{\bx, \by}(\bV_1) - F_{\bx, \by}(\bV_2)| \leq 2\delta^{-2} \left[\min\left\{ \|\bx\|_{\infty}, \|\by\|_{\infty}\right\} + \frac{1}{\sqrt{N}}\right]\|\bP_1 - \bP_2\|_F. \end{equation}
      Lastly, we bound
      \begin{align}
        \|\bP_1 - \bP_2\|_F
        &= \|\bV_1^\top\bV_1 - \bV_2^\top \bV_2\|_F \nonumber \\
        &= \|(\bV_1 - \bV_2)^\top\bV_1 + \bV_2^\top (\bV_1 - \bV_2)\|_F \nonumber \\
        &\leq \|(\bV_1 - \bV_2)^\top\bV_1\|_F + \|\bV_2^\top (\bV_1 - \bV_2)\|_F \nonumber \\
        &= \left(\Tr\left[(\bV_1 - \bV_2)^\top\bV_1\bV_1^\top(\bV_1 - \bV_2)\right]\right)^{1/2} \nonumber \\
        &\hspace{1cm} + \left(\Tr\left[(\bV_1 - \bV_2)^\top\bV_2\bV_2^\top(\bV_1 - \bV_2)\right]\right)^{1/2} \nonumber \\
        &= 2\|\bV_1 - \bV_2\|_F,
      \end{align}
      where we have used that $\bV_1\bV_1^\top = \bV_2\bV_2^\top = \bm I_r$, and the result follows.
  \end{proof}

    Therefore, and what is crucial to our argument, while for the worst-case $\bx \in B(\bm 0, 1)$, namely $\bx = \be_i$ a standard basis vector, $F_{\bx, \bx}$ will have Lipschitz constant $O(1)$, for typical $\bx \in B(\bm 0, 1)$, $F_{\bx, \bx}$ will rather have Lipschitz constant $\widetilde{O}(N^{-1/2})$.
    Moreover, the Lipschitz constant of $F_{\bx, \by}$ is comparable to the \emph{smaller} of the Lipschitz constants of $F_{\bx, \bx}$ and $F_{\by, \by}$.

    \begin{claim}
        \label{claim:exp}
        For $\bx, \by \in B(\bm 0, 1)$, $\EE_{\bV \sim \Haar(\mathsf{Stief}(N, r))} F_{\bx, \by}(\bV) = \la \bx, \by \ra + O_{\delta}(N^{-1})$.
    \end{claim}
    \begin{proof}
    We have
    \begin{align}
      \delta^2 \EE F_{\bx, \by}(\bV)
      &= \EE\Tr\left(\bV\bD_{\bx}\bV^\top\bV\bD_{\by}\bV^\top\right) \nonumber \\
      &\hspace{1cm} - \frac{1 - \sqrt{\delta}}{N} \la \one_N, \bx \ra \EE\Tr(\bV^\top \bV \bD_y) \nonumber \\
      &\hspace{1cm}- \frac{1 - \sqrt{\delta}}{N} \la \one_N, \by \ra \EE\Tr(\bV^\top \bV \bD_x) \nonumber \\
      &\hspace{1cm} + \frac{\delta(1 - \sqrt{\delta})^2}{N} \la \one_N, \bx \ra \la \one_N, \by \ra,
        \intertext{and by either the moment formulae of Proposition~\ref{prop:haar-moments} or an argument from orthogonal invariance of Haar measure, we have $\EE \bV^\top \bV = \delta \bm I_N$, whereby}
      &= \EE\Tr\left(\bV\bD_{\bx}\bV^\top\bV\bD_{\by}\bV^\top\right) - \frac{\delta(1 - \delta)}{N}\la \one_N, \bx \ra \la \one_N, \by \ra.
    \end{align}
    View $\bV \sim \Haar(\mathsf{Stief}(N, r))$ as the top $r \times N$ block of $\bQ \sim \Haar(\sO(N))$.
    Then, expanding the first term with the moment formulae of Proposition~\ref{prop:haar-moments},
    \begin{align}
      \hspace{0.5cm}&\hspace{-0.5cm}\EE\Tr\left(\bV\bD_{\bx}\bV^\top\bV\bD_{\by}\bV^\top\right) \nonumber \\
      &=  \sum_{i, j = 1}^N x_iy_j\left(\sum_{a, b = 1}^r\EE\left[Q_{ai}Q_{bi}Q_{aj}Q_{bj}\right]\right) \nonumber \\
            &= \sum_{i = 1}^N x_iy_i\left(\frac{3r}{N(N + 2)} + \frac{r(r - 1)}{N(N + 2)}\right) \nonumber \\
      &\hspace{1cm} + \sum_{1 \leq i < j \leq N}x_iy_j\left(\frac{r}{N(N + 2)} - \frac{r(r - 1)}{(N - 1)N(N + 2)}\right) \nonumber \\
            &= \frac{\delta}{N + 2}\left(r + 1 + \frac{r - 1}{N - 1}\right)\sum_{i = 1}^N x_iy_i \nonumber \\
      &\hspace{1cm} + \frac{\delta}{N + 2}\left(1 - \frac{r - 1}{N - 1}\right)\left(\sum_{i = 1}^N x_i\right)\left(\sum_{i = 1}^N y_i\right) \nonumber \\
      &= \left(\delta^2 + O_{\delta}(N^{-1})\right)\la \bx, \by \ra + \left(1 + O_{\delta}(N^{-1})\right)\frac{\delta(1 - \delta)}{N}\la \one_N, \bx \ra \la \one_N, \by \ra,
    \end{align}
    and since $|\la \bx, \by \ra| \leq \|\bx\|_2\|\by\|_2 \leq 1$ and $|\la \one_N, \bx\ra|\cdot |\la \one_N, \by \ra| \leq N$, the result follows.
\end{proof}

    \noindent
    Combining Claim~\ref{claim:lip}, Claim~\ref{claim:exp}, and the concentration result Proposition~\ref{prop:orth-lip-conc}, we find the following corollary on pointwise concentration of $F_{\bx, \by}(\bV)$.

    \begin{claim}
        \label{claim:conc}
        There exist constants $C_1, C_2 > 0$ depending only on $\delta$ such that, for any $\bx, \by \in B(\bm 0, 1)$,
        \begin{align}
          &\PP_{\bV \sim \Haar(\mathsf{Stief}(N, r))}\left[ \left|\la \bA_0^\orth \bx, \bA_0^\orth \by \ra - \la \bx, \by \ra\right| \geq \frac{C_1}{N} + t\right] \nonumber \\ &\hspace{1cm} \leq 2\exp\left(-\frac{C_2 N t^2}{(\min\{\|\bx\|_\infty, \|\by\|_\infty\} + N^{-1/2})^2}\right).
        \end{align}
    \end{claim}

        \noindent
    This concludes the first part of the argument.

    The remaining part of the argument is to apply a union bound of the probabilities controlled in Claim~\ref{claim:conc} over suitable nets of $B(\bm 0, 1)$.
    We divide our task into a bound over sparse vectors and vectors with bounded largest entry, very similar to the technique in \cite{Rudelson-inv,RV-inv} and especially \cite{Vershynin-products}.
    Introduce a parameter $\rho \in (0, 1)$ to be chosen later.
    Define
    \begin{align}
      B_s &\colonequals \left\{\by \in B(\bm 0, 1): \|\by\|_0 \leq \rho N \right\}, \\
      B_b &\colonequals \left\{\bz \in B(\bm 0, 1): \|\bz\|_{\infty} \leq \frac{1}{\sqrt{\rho N}}\right\}.
    \end{align}
    For any $\bx \in B(\bm 0, 1)$, we define $\by = \by(\bx)$ and $\bz = \bz(\bx)$ by thresholding the entries of $\bx$, setting $y_i \colonequals x_i \One\{|x_i| > \frac{1}{\sqrt{\rho N}}\}$ and $z_i \colonequals x_i \One\{|x_i| \leq \frac{1}{\sqrt{\rho N}}\}$.
    Then, $\bx = \by + \bz$, $\by \in B_s$, and $\bz \in B_b$.

    Introduce another parameter $\gamma \in (0, 1)$ to be chosen later.
    Let $N_s \subset B_s$ and $N_b \subset B_b$ be $\gamma$-nets.
    By a standard bound (see, e.g., Lemma 9.5 of \cite{LT-banach}), we may choose $|N_b| \leq \exp(2N / \gamma)$, and by the same bound applied to each choice of $\rho N$ support coordinates for an element of $B_s$, we may choose
    \begin{equation} |N_s| \leq \binom{N}{\lfloor \rho N \rfloor}\exp\left(\frac{2\rho N}{\gamma}\right) \leq \exp\left(\frac{2\rho N}{\gamma} + \rho N + \log\left(\frac{1}{\rho}\right)\rho N\right). \end{equation}
    To lighten the notation, let us set $\bS \colonequals \bA_0^{\orth^\top}\bA_0^\orth - \bm I_N$.
    The following is an adaptation to our setting of a standard technique for estimating a matrix norm over a net: we first bound
    \begin{align}
      \|\bS\|_{\op} &= \max_{\bx \in B(\bm 0, 1)}|\bx^\top \bS \bx| \nonumber \\
      &\leq \max_{\substack{\by \in B_s \\ \bz \in B_b }}|(\by + \bz)^\top \bS (\by + \bz)| \nonumber \\
      &\leq \max_{\by \in B_s} |\by^\top \bS \by| + \max_{\bz \in B_b} |\bz^\top \bS \bz| + 2\max_{\substack{\by \in B_s \\ \bz \in B_b}}|\by^\top \bS \bz | \nonumber \\
      &\leq \max_{\by \in N_s} |\by^\top \bS \by| + \max_{\bz \in N_b} |\bz^\top \bS \bz| + 2\max_{\substack{\by \in N_s \\ \bz \in N_b }}|\by^\top \bS \bz | + 12\gamma\|\bS\|_{\op}.
    \end{align}
    Rearranging this, we obtain
    \begin{align}
      \|\bS\|_{\op} &\leq \frac{1}{1 - 12\gamma}\left[\max_{\by \in N_s} |\by^\top \bS \by| + \max_{\bz \in N_b} |\bz^\top \bS \bz| + 2\max_{\substack{\by \in N_s \\ \bz \in N_b }}|\by^\top \bS \bz |\right].
    \end{align}
    Using Claim~\ref{claim:conc} and a union bound, we have that
    \begin{align}
      &\PP\left[\|\bS\|_{\op} \geq \frac{4}{1 - 12\gamma}\left(\frac{C_1}{N} + t\right)\right] \nonumber \\
      &\hspace{1cm} \leq 2(|N_b| + |N_s| \cdot |N_b|)\exp\left(-C_2 \frac{\rho}{(1 + \sqrt{\rho})^2} N^2 t^2\right) \nonumber \\
      &\hspace{2cm} + 2|N_s|\exp\left(-\frac{C_2}{2}Nt^2\right) \nonumber \\
      &\hspace{1cm} \leq 3\exp\left(N\left[\frac{2}{\gamma}(1 + \rho) + \rho + \log\left(\frac{1}{\rho}\right)\rho -C_2 \frac{\rho}{(1 + \sqrt{\rho})^2} N t^2\right]\right) \nonumber \\
      &\hspace{2cm} + 2\exp\left(N\left[\frac{2\rho}{\gamma} + \rho + \log\left(\frac{1}{\rho}\right)\rho-\frac{C_2}{2}t^2\right]\right).
    \end{align}
    Taking $\rho = N^{-1/2}$, $t = C_3 N^{-1/4}\log N$ for a large constant $C_3$, and $\gamma < \frac{1}{12}$ a small constant, we obtain the result.
\end{proof}

\subsection{Bounding the projection term \texorpdfstring{$\bmrob T^{(2)}$}{T2}: normalization}

\label{sec:normalization}

In this section, we show that the passing from the approximately normalized vectors $\delta^{-1/2}\bv_i$ discussed in the previous section to the exactly normalized vectors $\what{\bv}_i = \bv_i / \|\bv_i\|_2$ does not affect the construction of $\bT^{(2)}$ very much, as measured by operator norm.

\begin{lemma}
    \label{lem:normalization}
    Let $\bA^{\orth} \in \RR^{r(r + 1) / 2 \times N}$ have $\isovec(\delta^{-1}\bv_i\bv_i^\top - \frac{1}{r}\bm I_r)$ as columns, and let $\bA^{\norm} \in \RR^{r(r + 1) / 2 \times N}$ have $\isovec(\what{\bv}_i\what{\bv}_i^\top - \frac{1}{r}\bm I_r)$ as columns.
    Then, for all $K > 0$,
    \begin{equation} \PP\left[ \|\bA^\orth - \bA^\norm\|_{\op} \leq O_{\delta, K}\left(\frac{\log N}{N}\right)\right] \geq 1 - O_{\delta, K}(N^{-K}). \end{equation}
\end{lemma}
\begin{proof}
    Recall that $\bD \in \RR^{N \times N}$ is diagonal with $D_{ii} = \|\bv_i\|_2^2$.
    Then, we may write
    \begin{align}
      \bA^\norm
      &= \delta \bA^\orth \bD^{-1} + \frac{\delta}{r}\one_{\diag}\one_N^\top\bD^{-1} - \frac{1}{r}\one_{\diag}\one_N^\top \nonumber \\
      &= \bA^\orth + \bA^{\orth}(\delta \bD^{-1} - \bm I_N)+ \frac{1}{r}\one_{\diag}\one_N^\top(\delta \bD^{-1} - \bm I_N).
    \end{align}
    Therefore,
    \begin{align}
      \|\bA^\norm - \bA^\orth\|_{\op}
      &\leq \|\delta \bD^{-1} - \bm I_N\|_{\op}\left(\|\bA\|^{\orth} + \delta^{-1/2}\right) \nonumber \\
      &= \left(\max_{i \in [N]}\left|\frac{\delta}{D_{ii}} - 1\right|\right)\left(\|\bA\|^{\orth} + \delta^{-1/2}\right).
    \end{align}
    By Lemma~\ref{lem:orth-sym-cov} from the previous section, the second term is $O_{\delta}(1)$ with super-polynomially high probability.
    The result then follows by Proposition~\ref{prop:proj-entries}.
\end{proof}
\noindent
Next, we translate this result to the Gram matrix $\bA^{\norm^\top}\bA^\norm$.
\begin{corollary}
    \label{cor:gram-normalization}
    In the same setting as Lemma~\ref{lem:normalization}, for all $K > 0$,
    \begin{align} &\PP\left[ \|\bA^{\orth^\top}\bA^\orth - \bA^{\norm^\top}\bA^\norm\|_{\op} \leq O_{\delta, K}\left(\frac{\log N}{N}\right)\right] \nonumber \\
      &\hspace{1cm}\geq 1 -  O_{\delta, K}(N^{-K}). \end{align}
\end{corollary}
\begin{proof}
    We may bound
    \begin{align}
        &\|\bA^{\orth^\top}\bA^\orth - \bA^{\norm^\top}\bA^\norm\|_{\op} \nonumber \\
        &\hspace{1cm} = \|\bA^{\orth^\top}(\bA^\orth - \bA^\norm) +  (\bA^\orth - \bA^\norm)^\top\bA^{\norm}\|_{\op} \nonumber \\
        &\hspace{1cm} \leq \left(\|\bA^\orth\|_{\op} + \|\bA^{\norm}\|_{\op}\right)\|\bA^\orth - \bA^\norm\|_{\op} \nonumber \\
        &\hspace{1cm} \leq \left(2\|\bA^\orth\|_{\op} + \|\bA^\orth - \bA^{\norm}\|_{\op}\right)\|\bA^\orth - \bA^\norm\|_{\op}.
    \end{align}
    Then, combining Lemma~\ref{lem:orth-sym-cov} and Lemma~\ref{lem:normalization} gives the result.
\end{proof}

\subsection{Final main term bound: proof of Lemma~\ref{lem:main-term}}

We are now ready to complete the proof of Lemma~\ref{lem:main-term}.
First, we combine Lemma~\ref{lem:orth-sym-cov} and Corollary~\ref{cor:gram-normalization}.
Recall that these showed the following bounds, with high probability:
\begin{align}
  \|\bA^{\orth^\top}\bA^\orth\|_{\op} &\leq 1 + O_{\delta}\left(\frac{\log N}{N^{1/4}}\right), \tag{Lemma~\ref{lem:orth-sym-cov}} \\
  \|\bA^{\orth^\top}\bA^\orth - \bA^{\norm^\top}\bA^\norm\|_{\op} &\leq O_{\delta}\left(\frac{\log N}{N}\right). \tag{Corollary~\ref{cor:gram-normalization}}
\end{align}
Combining these, we obtain the following bound on $\bT^{(2)} = \bA^\norm\bA^{\norm^\top}$.
\begin{corollary}
    \label{cor:T2-bound}
    For all $\delta \in (0, 1)$,
    \begin{equation} \lim_{N \to \infty}\PP\left[\|\bT^{(2)}\|_{\op} \leq 1 + O_{\delta}\left(\frac{\log N}{N^{1/4}}\right)\right] = 1. \end{equation}
\end{corollary}

Next, we complete the proof of Lemma~\ref{lem:main-term}.
Recall that
\begin{equation} \widetilde{\bZ}^{(1a)} = \left(\alpha - \frac{2N}{r^2}\right)\one_{\diag}\one_{\diag}^\top + 2\bm I_{r(r + 1) / 2} + \bT^{(1)} - 2\bT^{(2)}, \end{equation}
and we have gathered the following bounds, holding with high probability:
\begin{align}
  |\bT^{(1)}| &\preceq O_{\delta}\left(\frac{\log N}{N}\right)\bm I_{r(r + 1) / 2} + \frac{2}{r}\one_{\diag}\one_{\diag}^\top \tag{Lemma \ref{lem:cross-term}} \\
  0 \preceq \bT^{(2)} &\preceq \left(1 + O_{\delta}\left(\frac{\log N}{N^{1/4}}\right)\right)\bm I_{r(r + 1) / 2} \tag{Corollary~\ref{cor:T2-bound}}
\end{align}
We can therefore control the minimum eigenvalue of $\widetilde{\bZ}^{(1a)}$ by, with high probability,
\begin{align}
  \widetilde{\bZ}^{(1a)} \succeq \left(\alpha - \frac{2N}{r^2} - \frac{2}{r}\right)\one_{\diag}\one_{\diag}^\top - O_{\delta}\left(\frac{\log N}{N^{1/4}}\right)\bm I_{r(r + 1) / 2}.
\end{align}
The first term is positive semidefinite for sufficiently large $N$ since $N / r^2 \to 0$, and thus we find
\begin{equation} \lim_{N \to \infty}\PP\left[\lambda_{\min}(\widetilde{\bZ}^{(1a)}) \geq -O_{\delta}\left(\frac{\log N}{N^{1/4}}\right)\right] = 1. \end{equation}

Finally, we must convert this to a bound on the smallest eigenvalue of $\bZ^{(1a)}$.
By Proposition~\ref{prop:Z1a-tZ1a}, letting $\bP = \bV^\top \bV$, we have
\begin{equation} \lambda_{\min}(\bZ^{(1a)}) \geq \min\left\{0, \frac{\lambda_{\min}(\widetilde{\bZ}^{(1a)})}{2\min_{i \in [N]} P_{ii}^2}\right\}. \end{equation}
By Proposition~\ref{prop:proj-entries}, with high probability $P_{ii} \geq \delta^{-1} - O_{\delta}(\sqrt{\frac{\log N}{N}})$ for all $i \in [N]$.
Substituting this above, we thus find the result of Lemma~\ref{lem:main-term},
\begin{equation} \lim_{N \to \infty}\PP\left[\lambda_{\min}(\bZ^{(1a)}) \geq -O_{\delta}\left(\frac{\log N}{N^{1/4}}\right)\right] = 1. \end{equation}

\section{Proof of positive semidefiniteness: correction term}

\label{sec:corr-term}

\begin{proof}[Proof of Lemma~\ref{lem:corr-term}]
    Recall that $\bm\Delta$ is non-zero only on pairs of index sets $\{i,j\}$ and $\{i,k\}$ that share an index.
    The non-zero entries are given by
    \begin{equation}
        \Delta_{\{i,j\}\{i,k\}} = \sum_{m = 1}^N M_{im}^2M_{jm}M_{km} - M_{ij}M_{ik}.
    \end{equation}

    Let us compute the quadratic form of $\bm\Delta$ with $\ba \in \RR^{N(N - 1) / 2}$, which we view as having entries $a_{\{i,j\}} = A_{ij}$ for some $\bA \in \RR^{r \times r}_{\sym}$ with $\diag(\bm A) = \bm 0$ (i.e.\ $\ba = \offdiag(\bA)$).
    We then have, expanding with a correction for double-counting the diagonal terms,
    \begin{align}
      \ba^\top \bm\Delta \ba
      &= \sum_{i = 1}^N \sum_{j = 1}^N\sum_{k = 1}^NA_{ij}\left\{\left(\sum_{m = 1}^N M_{im}^2M_{jm}M_{km}\right) - M_{ij}M_{ik}\right\}A_{ik} \nonumber \\
      &\hspace{1cm} - \sum_{1 \leq i < j \leq N} A_{ij}^2\left\{\left(\sum_{m = 1}^NM_{im}^2M_{jm}^2\right) - M_{ij}^2\right\} \nonumber \\
      &= \sum_{i = 1}^N \left\{\left(\sum_{m = 1}^NM_{im}^2(\bA \bM)_{im}^2\right) - (\bA \bM)_{ii}^2\right\}  \nonumber \\
      &\hspace{1cm} - \sum_{1 \leq i < j \leq N}A_{ij}^2((\bM^{\circ 2})^2 - \bM^{\circ 2})_{ij}) \nonumber \\
      &= \la \bM^{\circ 2}, (\bA\bM)^{\circ 2} \ra - \Tr((\bA\bM)^{\circ 2}) - \la \bA^{\circ 2}, (\bM^{\circ 2})^2 - \bM^{\circ 2} \ra  \nonumber \\
      &= \la \bM^{\circ 2} - \bm I_N, (\bA\bM)^{\circ 2} \ra - \la \bA^{\circ 2}, (\bM^{\circ 2})^2 - \bM^{\circ 2} \ra.
        \intertext{Following the notation of Sections~\ref{sec:unnormalized} and \ref{sec:normalization}, we let $\bR \colonequals \bA^{\norm^\top}\bA^\norm$ be the Gram matrix of $\isovec(\what{\bv}_i\what{\bv}_i^\top - \frac{1}{r}\bm I_r)$, whereby $\bM^{\circ 2} = \bR + \frac{1}{r}\one_N\one_N^\top$, allowing us to continue by expanding the second term}
      &= \la \bM^{\circ 2} - \bm I_N, (\bA\bM)^{\circ 2} \ra - \frac{1}{r^2}(\one_N^\top \bR\one_N^\top)\|\bA\|_F^2 + \frac{1}{r}\|\bA\|_F^2 \nonumber \\
      &\hspace{1cm} - \left\la \bA^{\circ 2}, \bR^2 - \bR +  \frac{1}{r}\one_N\one_N^\top \bR + \frac{1}{r}\bR\one_N\one_N^\top  \right\ra.
    \end{align}

    We will now use the following inequality that allows us to bound inner products with the Schur square of a matrix by its Frobenius norm:
    \begin{equation}
      |\la \bX^{\circ 2}, \bY \ra| \leq \sum_{i, j}|Y_{ij}|X_{ij}^2 \leq \|\bY\|_{\ell^\infty}\|\bX\|_F^2 \leq \|\bY\|_{\op}\|\bX\|_F^2.
    \end{equation}
    (We denote by $\|\bY\|_{\ell^\infty}$ the vectorized supremum norm.)
    This will result in one term involving $\|\bA\bM\|_F^2$, which we bound by
    \begin{equation}
      \|\bA\bM\|_F^2 = \Tr(\bA^\top\bA\bM^2) \leq \|\bM\|_{\op}^2\Tr(\bA^\top \bA) = \|\bM\|_{\op}^2\|\bA\|_F^2.
    \end{equation}
    Combining these inequalities and noting that $\|\ba\|_2^2 = \frac{1}{2}\|\bA\|_F^2$, we find (for $\ba \neq \bm 0$) that
    \begin{align}
      \frac{|\ba^\top \bm\Delta \ba|}{\|\ba\|_2^2}
      &\leq 2\bigg(\|\bM^{\circ 2} - \bm I_N\|_{\ell^\infty}\|\bM\|_{\op}^2 + \frac{1}{r} + \frac{\one_N^\top\bR\one_N}{r^2} \nonumber \\
      &\hspace{2cm} + \|\bR^2 - \bR\|_{\op} + \frac{2\sqrt{N}}{r}\|\bR\one_N\|_2\bigg).
    \end{align}
    Since $\diag(\bM) = \one$, we have by Corollary~\ref{cor:M-entries} that, with high probability,
    \begin{equation} \|\bM^{\circ 2} - \bm I_N\|_{\ell^\infty} = \max_{\substack{i, j \in [N] \\ i \neq j}} |M_{ij}|^2 \leq O_{\delta}\left(\frac{\log N}{N}\right), \end{equation}
    and by Corollary~\ref{cor:M-norm}, $\|\bM\|_{\op} = O_{\delta}(1)$ with high probability.

    To control the terms involving $\bR$, recall that from Lemmata~\ref{lem:orth-sym-cov} and \ref{lem:normalization} it follows that with high probability
    \begin{equation} \|\bR - \bP_{\one_N^\perp}\|_{\op} \leq O_{\delta}\left(\frac{\log N}{N^{1/4}}\right). \end{equation}
    Thus we have, with high probability,
    \begin{align}
      \frac{|\one_N^\top \bR \one_N|}{r^2} &\leq \frac{N}{r^2}\|\bR - \bP_{\one_N^\perp}\|_{\op} = O_{\delta}\left(\frac{\log N}{N^{5/4}}\right), \\
      \|\bR^2 - \bR\|_{\op} &\leq 3\|\bR - \bP_{\one_N^\perp}\|_{\op} + \|\bR - \bP_{\one_N^\perp}\|_{\op}^2 = O_{\delta}\left(\frac{\log ^2N}{N^{1/4}}\right), \\
      \frac{2\sqrt{N}}{r}\|\bR\one_N\|_2 &\leq \frac{2N}{r}\|\bR - \bP_{\one_N^\perp}\|_{\op} = O_{\delta}\left(\frac{\log N}{N^{1/4}}\right).
    \end{align}
    Combining these results, Lemma~\ref{lem:corr-term} follows.
\end{proof}

\begin{remark}
    We outline a simpler argument for Lemma~\ref{lem:corr-term} which suggests a sharper estimate, though it seems more difficult to formalize due to the dependency structure of $\bM$.
    Since $\bm\Delta$ is sparse (with non-zero entries only when the row and column index sets $\{i,j\}$ and $\{k,\ell\}$ share an element), our strategy will be to apply the Gershgorin circle theorem.
    Thus we must bound the diagonal and off-diagonal entries of $\bm \Delta$.

    The summation defining the diagonal entries contains only positive summands, so we have
    \begin{equation}
        \Delta_{\{i,j\}\{i,j\}} = \sum_{\substack{m = 1 \\ m \neq i}}^N M_{im}^2 M_{jm}^2 = M_{ij}^2 + \sum_{\substack{m = 1 \\ m \neq i, j}}^N M_{im}^2 M_{jm}^2 = \widetilde{O}\left(N^{-1}\right),
    \end{equation}
    using Proposition~\ref{prop:proj-entries} to control each term of the sum.
    (We omit bounds expressing ``with high probability'' statements and indulge in the logarithm-concealing $\widetilde{O}$ notation in this informal discussion.)

    The off-diagonal entries of $\bm\Delta$ are more difficult to control.
    They are
    \begin{equation}
        \Delta_{\{i,j\}\{i,k\}} = \sum_{\substack{m = 1 \\ m \neq i}}^N M_{im}^2 M_{jm}M_{km} = \sum_{\substack{m = 1 \\ m \neq i}}^N M_{im}^2 M_{jm}M_{km}.
    \end{equation}
    Here, we must capture the cancellations due to random signs: a naive application of the triangle inequality would give $|\Delta_{\{i,j\}\{i,k\}}| = \widetilde{O}(N^{-1})$, which would be insufficient for the Gershgorin circle theorem argument since there are $\Omega(N)$ non-zero entries in each row of $\bm\Delta$.
    On the other hand, a scaling like that in the central limit theorem with independent signs would give $|\Delta_{\{i,j\}\{i,k\}}| = \widetilde{O}(N^{-3/2})$, which would suffice, and would give $\|\bm\Delta\|_{\op} = \widetilde{O}(N^{-1/2})$, stronger by a factor of $N^{1/4}$ than the result of Lemma~\ref{lem:corr-term} (up to logarithmic factors).
\end{remark}

\section*{Acknowledgements}
\addcontentsline{toc}{section}{Acknowledgements}

We thank Jess Banks, Ankur Moitra, Andrea Montanari, Cristopher Moore, Tselil Schramm, and Ramon van Handel for useful discussions.
We also thank the authors of \cite{MRX-2019-Lifting2To4} for generously providing an early version of their manuscript.

\newpage

\addcontentsline{toc}{section}{References}
\bibliographystyle{alpha}
\bibliography{main}

\newcommand{\etalchar}[1]{$^{#1}$}
\begin{thebibliography}{STDHJ07}

\bibitem[ABM18]{ABM-rem}
Louigi Addario-Berry and Pascal Maillard.
\newblock The algorithmic hardness threshold for continuous random energy
  models.
\newblock {\em arXiv preprint arXiv:1810.05129}, 2018.

\bibitem[BHK{\etalchar{+}}19]{BHKKMP-pc}
Boaz Barak, Samuel Hopkins, Jonathan Kelner, Pravesh~K Kothari, Ankur Moitra,
  and Aaron Potechin.
\newblock A nearly tight sum-of-squares lower bound for the planted clique
  problem.
\newblock {\em SIAM Journal on Computing}, 48(2):687--735, 2019.

\bibitem[BK18]{BK-ETF}
Afonso~S Bandeira and Dmitriy Kunisky.
\newblock A {Gramian} description of the degree 4 generalized elliptope.
\newblock {\em arXiv preprint arXiv:1812.11583}, 2018.

\bibitem[BK19]{BK-sampta}
Afonso~S Bandeira and Dmitriy Kunisky.
\newblock Sum-of-squares optimization and the sparsity structure of equiangular
  tight frames.
\newblock In {\em 2019 International Conference on Sampling Theory and
  Applications (SampTA 2019)}. IEEE, 2019.

\bibitem[BKW20]{BKW-LDLR}
Afonso~S Bandeira, Dmitriy Kunisky, and Alexander~S Wein.
\newblock Computational hardness of certifying bounds on constrained {PCA}
  problems.
\newblock In Thomas Vidick, editor, {\em 11th Innovations in Theoretical
  Computer Science Conference (ITCS 2020)}, volume 151 of {\em Leibniz
  International Proceedings in Informatics (LIPIcs)}, pages 78:1--78:29,
  Dagstuhl, Germany, 2020. Schloss Dagstuhl--Leibniz-Zentrum fuer Informatik.

\bibitem[BPT12]{BPT-book}
Grigoriy Blekherman, Pablo~A Parrilo, and Rekha~R Thomas.
\newblock {\em Semidefinite optimization and convex algebraic geometry}.
\newblock SIAM, 2012.

\bibitem[BS14]{BS-SOS}
Boaz Barak and David Steurer.
\newblock Sum-of-squares proofs and the quest toward optimal algorithms.
\newblock {\em arXiv preprint arXiv:1404.5236}, 2014.

\bibitem[CM08]{CM-haar}
Sourav Chatterjee and Elizabeth Meckes.
\newblock Multivariate normal approximation using exchangeable pairs.
\newblock {\em Alea}, 4:257--283, 2008.

\bibitem[CRT08]{Casazza-ETF}
Peter~G Casazza, Dan Redmond, and Janet~C Tremain.
\newblock Real equiangular frames.
\newblock In {\em 42nd Annual Conference on Information Sciences and Systems
  (CISS 2008)}, pages 715--720. IEEE, 2008.

\bibitem[DL09]{DL-cut}
Michel~Marie Deza and Monique Laurent.
\newblock {\em Geometry of cuts and metrics}, volume~15.
\newblock Springer, 2009.

\bibitem[DMS{\etalchar{+}}17]{DMS-cut}
Amir Dembo, Andrea Montanari, Subhabrata Sen, et~al.
\newblock Extremal cuts of sparse random graphs.
\newblock {\em Annals of Probability}, 45(2):1190--1217, 2017.

\bibitem[FM15]{FM-ETF}
Matthew Fickus and Dustin~G Mixon.
\newblock Tables of the existence of equiangular tight frames.
\newblock {\em arXiv preprint arXiv:1504.00253}, 2015.

\bibitem[FSP16]{Fawzi-SOS}
Hamza Fawzi, James Saunderson, and Pablo~A Parrilo.
\newblock Sparse sums of squares on finite abelian groups and improved
  semidefinite lifts.
\newblock {\em Mathematical Programming}, 160(1-2):149--191, 2016.

\bibitem[Hop18]{sam-thesis}
Samuel Hopkins.
\newblock {\em Statistical Inference and the Sum of Squares Method}.
\newblock PhD thesis, Cornell University, 2018.

\bibitem[HS17]{HS-bayesian}
Samuel~B Hopkins and David Steurer.
\newblock Efficient {Bayesian} estimation from few samples: community detection
  and related problems.
\newblock In {\em 58th Annual Symposium on Foundations of Computer Science
  (FOCS 2017)}, pages 379--390. IEEE, 2017.

\bibitem[JKR19]{JRK-SOS}
Vishesh Jain, Frederic Koehler, and Andrej Risteski.
\newblock Mean-field approximation, convex hierarchies, and the optimality of
  correlation rounding: a unified perspective.
\newblock In {\em 51st Annual ACM SIGACT Symposium on Theory of Computing (STOC
  2019)}, pages 1226--1236. ACM, 2019.

\bibitem[Kar72]{Karp-NP}
Richard~M Karp.
\newblock Reducibility among combinatorial problems.
\newblock In {\em Complexity of computer computations}, pages 85--103.
  Springer, 1972.

\bibitem[KLM16]{KLM-2016-SymmetricFormulations}
Adam Kurpisz, Samuli Lepp{\"a}nen, and Monaldo Mastrolilli.
\newblock Sum-of-squares hierarchy lower bounds for symmetric formulations.
\newblock In {\em International Conference on Integer Programming and
  Combinatorial Optimization}, pages 362--374. Springer, 2016.

\bibitem[KWB19]{low-deg-notes}
Dmitriy Kunisky, Alexander~S Wein, and Afonso~S Bandeira.
\newblock Notes on computational hardness of hypothesis testing: Predictions
  using the low-degree likelihood ratio.
\newblock {\em arXiv preprint arXiv:1907.11636}, 2019.

\bibitem[Las01]{Lasserre-SOS-survey}
Jean~B Lasserre.
\newblock Global optimization with polynomials and the problem of moments.
\newblock {\em SIAM Journal on Optimization}, 11(3):796--817, 2001.

\bibitem[Lau03]{Laurent-SOS}
Monique Laurent.
\newblock Lower bound for the number of iterations in semidefinite hierarchies
  for the cut polytope.
\newblock {\em Mathematics of operations research}, 28(4):871--883, 2003.

\bibitem[Lau09]{Laurent-SOS-survey}
Monique Laurent.
\newblock Sums of squares, moment matrices and optimization over polynomials.
\newblock In {\em Emerging applications of algebraic geometry}, pages 157--270.
  Springer, 2009.

\bibitem[Led01]{Ledoux-conc}
Michel Ledoux.
\newblock {\em The concentration of measure phenomenon}.
\newblock Number~89 in Mathematical Surveys \& Monographs. American
  Mathematical Society, 2001.

\bibitem[LM00]{LM-chi-squared}
Beatrice Laurent and Pascal Massart.
\newblock Adaptive estimation of a quadratic functional by model selection.
\newblock {\em Annals of Statistics}, pages 1302--1338, 2000.

\bibitem[LT94]{li:94}
Chi-Kwong Li and Bit-Shun Tam.
\newblock A note on extreme correlation matrices.
\newblock {\em SIAM Journal on Matrix Analysis and Applications},
  15(3):903--908, 1994.

\bibitem[LT13]{LT-banach}
Michel Ledoux and Michel Talagrand.
\newblock {\em Probability in Banach spaces: isoperimetry and processes}.
\newblock Springer Science \& Business Media, 2013.

\bibitem[Mon19]{Montanari-SK}
Andrea Montanari.
\newblock Optimization of the {Sherrington-Kirkpatrick Hamiltonian}.
\newblock In {\em 60th Annual Symposium on Foundations of Computer Science
  (FOCS 2019)}, pages 1417--1433. IEEE, 2019.

\bibitem[MRX19]{MRX-2019-Lifting2To4}
Sidhanth Mohanty, Prasad Raghavendra, and Jeff Xu.
\newblock Lifting sum-of-squares lower bounds: degree-2 to degree-4.
\newblock {\em arXiv preprint arXiv:1911.01411}, 2019.

\bibitem[MS16]{MS-deg2}
Andrea Montanari and Subhabrata Sen.
\newblock Semidefinite programs on sparse random graphs and their application
  to community detection.
\newblock In {\em 48th Annual ACM Symposium on Theory of Computing (STOC
  2016)}, pages 814--827. ACM, 2016.

\bibitem[O'D17]{ODonnell-2017-SOSNotAutomatizable}
Ryan O'Donnell.
\newblock {SOS} is not obviously automatizable, even approximately.
\newblock In {\em 8th Innovations in Theoretical Computer Science Conference
  (ITCS 2017)}. Schloss Dagstuhl-Leibniz-Zentrum fuer Informatik, 2017.

\bibitem[Pan13a]{Panchenko-ultra}
Dmitry Panchenko.
\newblock The {Parisi} ultrametricity conjecture.
\newblock {\em Annals of Mathematics}, 177(1):383--393, 2013.

\bibitem[Pan13b]{Panchenko-SK}
Dmitry Panchenko.
\newblock {\em The Sherrington-Kirkpatrick model}.
\newblock Springer Science \& Business Media, 2013.

\bibitem[Par79]{Parisi-SK}
Giorgio Parisi.
\newblock Infinite number of order parameters for spin-glasses.
\newblock {\em Physical Review Letters}, 43(23):1754, 1979.

\bibitem[RM14]{MR-tensor}
Emile Richard and Andrea Montanari.
\newblock A statistical model for tensor {PCA}.
\newblock In {\em Advances in Neural Information Processing Systems}, pages
  2897--2905, 2014.

\bibitem[Rud08]{Rudelson-inv}
Mark Rudelson.
\newblock Invertibility of random matrices: norm of the inverse.
\newblock {\em Annals of Mathematics}, pages 575--600, 2008.

\bibitem[RV08]{RV-inv}
Mark Rudelson and Roman Vershynin.
\newblock The {Littlewood}--{Offord} problem and invertibility of random
  matrices.
\newblock {\em Advances in Mathematics}, 218(2):600--633, 2008.

\bibitem[RW17]{RW-2017-BitComplexity}
Prasad Raghavendra and Benjamin Weitz.
\newblock On the bit complexity of sum-of-squares proofs.
\newblock In Ioannis Chatzigiannakis, Piotr Indyk, Fabian Kuhn, and Anca
  Muscholl, editors, {\em 44th International Colloquium on Automata, Languages,
  and Programming (ICALP 2017)}, volume~80 of {\em Leibniz International
  Proceedings in Informatics (LIPIcs)}, pages 80:1--80:13, Dagstuhl, Germany,
  2017. Schloss Dagstuhl--Leibniz-Zentrum fuer Informatik.

\bibitem[SK75]{SK}
David Sherrington and Scott Kirkpatrick.
\newblock Solvable model of a spin-glass.
\newblock {\em Physical review letters}, 35(26):1792, 1975.

\bibitem[SKM89]{SKM-tensor}
PK~Suetin, Alexandra~I Kostrikin, and Yu~I Manin.
\newblock {\em Linear algebra and geometry}.
\newblock CRC Press, 1989.

\bibitem[STDHJ07]{Tropp-ETF}
M{\'a}ty{\'a}s~A Sustik, Joel~A Tropp, Inderjit~S Dhillon, and Robert~W
  Heath~Jr.
\newblock On the existence of equiangular tight frames.
\newblock {\em Linear Algebra and its Applications}, 426(2-3):619--635, 2007.

\bibitem[Sub18]{Subag-frsb}
Eliran Subag.
\newblock Following the ground-states of full-{RSB} spherical spin glasses.
\newblock {\em arXiv preprint arXiv:1812.04588}, 2018.

\bibitem[Tal06]{Talagrand-Parisi}
Michel Talagrand.
\newblock The {Parisi} formula.
\newblock {\em Annals of Mathematics}, pages 221--263, 2006.

\bibitem[Ver11]{Vershynin-products}
Roman Vershynin.
\newblock Spectral norm of products of random and deterministic matrices.
\newblock {\em Probability Theory and Related Fields}, 150(3-4):471--509, 2011.

\end{thebibliography}

\end{document}